\documentclass{llncs}
\usepackage{a4wide}

\usepackage{xcolor,pgf,tikz,pgflibraryarrows,pgffor,pgflibrarysnakes}
\usetikzlibrary{automata,shapes,snakes,backgrounds,decorations.pathreplacing}
\usepackage{bbold}
\usepackage{lscape}

\usepackage{amsmath}
\newtheorem{assumption}{Assumption}
\usepackage{amssymb,stmaryrd}
\usepackage{algorithm}
\usepackage[noend]{algorithmic}
\usepackage{comment}
\usepackage{times}
\usepackage{url}
\usepackage{wrapfig}
\usepackage{placeins}
\usepackage{enumerate}
\usepackage{algorithm}
\usepackage{algorithmic}
\title{Symblicit algorithms for optimal strategy synthesis in monotonic Markov decision processes (extended version)\thanks{This work has been partly supported by ERC Starting Grant (279499: inVEST), ARC project (number AUWB-2010-10/15-UMONS-3) and European project Cassting (FP7-ICT-601148).}}

\author{Aaron Bohy$^1$ \and V\'eronique Bruy\`ere$^1$ \and Jean-Fran\c{c}ois Raskin$^{2}$}	 
\institute{$^1$Universit\'e de Mons$\quad$ $^2$Universit\'e Libre de Bruxelles}

\newcommand{\prism}{\ensuremath{\sf{PRISM}}}
\newcommand{\modest}{\ensuremath{\sf{MODEST}}}
\newcommand{\mrmc}{\ensuremath{\sf{MRMC}}}
\newcommand{\quasy}{\ensuremath{\sf{QUASY}}}

\newcommand{\defeq}{\stackrel{\sf{def}}{=}}

\newcommand{\nat}{{\mathbb{N}}}
\definecolor{light-gray}{gray}{0.9}
\definecolor{light-gray2}{gray}{0.96}

\newcommand{\CTL}{\ensuremath{\sf{CTL}} }
\newcommand{\PCTL}{\ensuremath{\sf{PCTL}} }
\newcommand{\LTL}{\ensuremath{\sf{LTL}} }
\newcommand{\LTLMP}{\textnormal{\ensuremath{\sf{LTL_\textsf{MP}}}} }
\newcommand{\Lit}{\textnormal{\textsf{Lit}} }

\newcommand{\MP}{\mathsf{MP}}
\newcommand{\TS}{\mathsf{TS}}
\newcommand{\mpval}{\mathsf{Val}}
\newcommand{\Outcome}{{\sf{Out}}}

\newcommand{\Win}{{\sf{Win}}}
\newcommand{\nextLTL}{{\sf{X}}}
\newcommand{\until}{{\sf{U}}}

\newcommand{\probI}{\pi_I}

\newcommand{\R}{\mathbb{R}}

\newcommand{\Z}{\mathbb{Z}}

\newcommand{\distrset}{\mathcal{D}}
\newcommand{\support}{\mathsf{Supp}}
\newcommand{\domain}{\mathsf{Dom}}
\newcommand{\ActionsO}{\Sigma}
\newcommand{\ActionsI}{T}
\newcommand{\actionO}{\sigma}
\newcommand{\actionI}{\tau}
\newcommand{\edges}{\textnormal{\textbf{E}}}
\newcommand{\distr}{\textnormal{\textbf{D}}}
\newcommand{\probmat}{\textnormal{\textbf{P}}}
\newcommand{\suc}{\mathsf{succ}}
\newcommand{\id}{\textnormal{\textbf{I}}}
\newcommand{\totfunc}{\mathcal{F}_\textnormal{tot}}
\newcommand{\enabledactions}{\ActionsO_s}
\newcommand{\enabledactionsprime}{\ActionsO_{s'}}
\newcommand{\enabledactionsxi}{\ActionsO_{x_i}}
\newcommand{\enabledstates}{S_\actionO}
\newcommand{\reward}{\textnormal{\textbf{C}}}
\newcommand{\ETP}{\mathbb{E}^{\textnormal{TS}_G}}
\newcommand{\EMP}{\mathbb{E}^\textnormal{MP}}
\newcommand{\Edot}{\mathbb{E}^{~\cdot}}

\newcommand{\symbMDP}{\mathcal{M}}
\newcommand{\symbReward}{\mathcal{C}}
\newcommand{\symbMC}{\symbMDP_{\lambda_n}}
\newcommand{\symbRewardMC}{\symbReward_{\lambda_n}}
\newcommand{\symbQuotient}{\symbMC'}
\newcommand{\symbRewardQuotient}{\symbRewardMC'}

\newcommand{\symbValue}{\mathcal{X}}
\newcommand{\symbGoal}{\mathcal{G}}
\newcommand{\explQuotient}{M_{\lambda_n}'}
\newcommand{\explRewardQuotient}{\reward_{\lambda_n}'}
\newcommand{\equivlump}{\sim_L}

\newcommand{\diff}{\backslash}
\newcommand{\pseudclos}{\updownarrow\!\!}
\newcommand{\antclos}{\downarrow\!\!}
\newcommand{\antunion}{\dot\cup}
\newcommand{\antinter}{\dot\cap}
\newcommand{\PA}{PA}

\newcommand{\monMDP}{M_\preceq}
\newcommand{\monMC}{M_{\preceq,\lambda}}
\newcommand{\monMCzero}{M_{\preceq,\lambda_0}}
\newcommand{\PreMC}{\Pre_\lambda}
\newcommand{\Pre}{\textnormal{\textsf{Pre}}}

\newcommand{\apre}{\mathsf{APre}}

\newcommand{\val}{block}

\newcommand{\equivdistr}{\sim_{\distr,\lambda}}
\newcommand{\partdistr}{S_{\equivdistr}}
\newcommand{\equivreward}{\sim_{\reward,\lambda}}
\newcommand{\partreward}{S_{\equivreward}}

\newcommand{\equivsigmarew}{\sim_{\reward,\actionO}}
\newcommand{\equivsigmaprob}{\sim_{\probmat,\actionO}}

\newcommand{\equivstrat}{\sim_\lambda}
\newcommand{\equivstratzero}{\sim_{\lambda_0}}
\newcommand{\equivstratp}{\sim_{\lambda'}}
\newcommand{\partstrat}{S_{\equivstrat}}

\newcommand{\partlump}{S_{\equivlump}}

\newcommand{\equiva}{\sim_{l_\actionO}}
\newcommand{\parta}{(S_{\actionO})_{\equiva}}
\newcommand{\equivc}{\sim_{\probmat,\actionO,C}}
\newcommand{\partc}{S_{\equivc}}

\newcommand{\equivalump}{\sim_{l_\actionO \wedge L}}
\newcommand{\partalump}{(S_\actionO)_{\equivalump}}
\newcommand{\partprom}{(S_\actionO)_{\equivalump}^<}

\newcommand{\laB}{l_\actionO(B)}
\newcommand{\las}{l_\actionO(s)}
\newcommand{\lasp}{l_\actionO(s')}
\newcommand{\lac}{l_\actionO(C)}
\newcommand{\lad}{l_\actionO(D)}

\newcommand{\timeout}{\textbf{TO}}
\newcommand{\memout}{\textbf{MO}}

\begin{document}
  \maketitle	

\begin{abstract} 
When treating Markov decision processes (MDPs) with large state spaces, using explicit representations quickly becomes unfeasible. Lately, Wimmer et al. have proposed a so-called symblicit algorithm for the synthesis of optimal strategies in MDPs, in the quantitative setting of expected mean-payoff. This algorithm, based on the strategy iteration algorithm of Howard and Veinott, efficiently combines symbolic and explicit data structures, and uses binary decision diagrams as symbolic representation. The aim of this paper is to show that the new data structure of pseudo-antichains (an extension of antichains) provides another interesting alternative, especially for the class of monotonic MDPs. We design efficient pseudo-antichain based symblicit algorithms (with open source implementations) for two quantitative settings: the expected mean-payoff and the stochastic shortest path. For two practical applications coming from automated planning and \LTL synthesis, we report promising experimental results w.r.t. both the run time and the memory consumption.
\end{abstract}

\section{Introduction}
Markov decision processes~\cite{putermanMDP,DBLP:books/daglib/0020348} (MDPs) are rich models that exhibit both nondeterministic choices and stochastic transitions. Model-checking and synthesis algorithms for MDPs exist for logical properties expressible in the logic \PCTL\cite{DBLP:journals/fac/HanssonJ94}, a stochastic extension of \CTL\cite{DBLP:conf/lop/ClarkeE81}, and are implemented in tools like \prism~\cite{KNP11}, \modest~\cite{DBLP:conf/fdl/Hartmanns12}, \mrmc~\cite{DBLP:journals/pe/KatoenZHHJ11}\dots There also exist algorithms for {\em quantitative properties} such as the long-run average reward (mean-payoff) or the stochastic shortest path, that have been implemented in tools like \quasy~\cite{DBLP:conf/tacas/ChatterjeeHJS11} and \prism~\cite{DBLP:conf/vmcai/EssenJ12}.


There are two main families of algorithms for MDPs. First, {\em value iteration} algorithms assign values to states of the MDPs and refines locally those values by successive approximations. If a fixpoint is reached, the value at a state $s$ represents a probability or an expectation that can be achieved by an optimal strategy that resolves the choices present in the MDP starting from $s$. This value can be, for example, the maximal probability to reach a set of goal states.  Second, {\em strategy iteration} algorithms start from an arbitrary strategy and iteratively improve the current strategy by local changes up to the convergence to an optimal strategy. Both methods have their advantages and disadvantages. Value iteration algorithms usually lead to easy and efficient implementations, but in general the fixpoint is not guaranteed to be reached in a finite number of iterations, and so only approximations are computed. On the other hand, strategy iteration algorithms have better theoretical properties as convergence towards an optimal strategy in a finite number of steps is usually ensured, but they often require to solve systems of linear equations, and so they are more difficult to implement efficiently.


When considering large MDPs, that are obtained from high level descriptions or as the product of several components, explicit methods often exhaust available memory and are thus impractical. This is the manifestation of the well-known {\em state explosion problem}. In non-probabilistic systems, symbolic data structures such as binary decision diagrams (BDDs) have been investigated~\cite{DBLP:journals/iandc/BurchCMDH92} to mitigate this phenomenon. For probabilistic systems, multi-terminal BDDs  (MTBDDs) are useful but they are usually limited to systems with around $10^{10}$ or $10^{11}$ states only~\cite{prismwebsite}. Also, as mentioned above, some algorithms for MDPs rely on solving linear systems, and there is no easy use of BDD like structures for implementing such algorithms.


Recently, Wimmer et al.~\cite{DBLP:conf/qest/WimmerBBHCHDT10} have proposed a method that {\em mixes} symbolic and explicit representations to efficiently implement the Howard and Veinott strategy iteration algorithm~\cite{howard1960dynamic,veinott1966finding} to synthesize optimal strategies for mean-payoff objectives in MDPs. Their solution is as follows. First, the MDP is represented and handled symbolically using MTBDDs. Second, a strategy is fixed symbolically and the MDP is transformed into a Markov chain (MC). To analyze this MC, a linear system needs to be constructed from its state space. As this state space is potentially huge, the MC is first reduced by {\em lumping}~\cite{kemeny1960jl,buchholz1994exact} (bisimulation reduction), and then a (hopefully) compact linear system can be constructed and solved. Solutions to this linear system allow to show that the current strategy is optimal, or to obtain sufficient information to improve it. A new iteration is then started. The main difference between this method and the other methods proposed in the literature is its {\em hybrid nature}: it is symbolic for handling the MDP and for computing the lumping, and it is explicit for the analysis of the reduced MC. This is why the authors of~\cite{DBLP:conf/qest/WimmerBBHCHDT10} have coined their approach {\em symblicit}.

\paragraph{Contributions.} In this paper, we build on the symblicit approach described above. Our contributions are threefold. 
First, we show that the symblicit approach and strategy iteration can also be efficiently applied to the {\em stochastic shortest path} problem. We start from an algorithm proposed by Bertsekas and Tsitsiklis~\cite{bertsekas1996neuro} with a preliminary step of de Alfaro~\cite{DBLP:conf/concur/Alfaro99}, and we show how to cast it in the symblicit approach.
Second, we show that alternative data structures can be more efficient than BDDs or MTBDDs for implementing a symblicit approach, both for mean-payoff and stochastic shortest path objectives. In particular, we consider a natural class of MDPs with {\em monotonic properties} on which our alternative data structure is more efficient. For such MDPs, as for subset constructions in automata theory~\cite{DBLP:conf/cav/WulfDHR06,DBLP:conf/tacas/DoyenR07}, antichain based data structures usually behave better than BDDs. The application of antichains to monotonic MDPs requires nontrivial extensions: for instance, to handle the lumping step, we need to generalize existing antichain based data structures in order to be closed under negation. To this end, we introduce a new data structure called {\em pseudo-antichain}.
Third, we have implemented our algorithms and we show that they are more efficient than existing solutions on natural examples of monotonic MDPs. We show that monotonic MDPs naturally arise in probabilistic planning~\cite{blum2000probabilistic} and when optimizing controllers synthesized from \LTL specifications with mean-payoff objectives~\cite{DBLP:conf/tacas/BohyBFR13}.

\paragraph{Structure of the paper.} In Section~\ref{sec:prelim}, we recall the useful definitions, and introduce the notion of monotonic MDP. In Section~\ref{sec:si}, we recall strategy iteration algorithms for mean-payoff and stochastic shortest path objectives, and we present the symblicit version of those algorithms. We introduce the notion of pseudo-antichains in Section~\ref{sec:pa}, and we describe our pseudo-antichain based symblicit algorithms in Section~\ref{sec:paalgo}. In Section~\ref{sec:experiments}, we propose two applications of the symblicit algorithms and give experimental results. Finally in Section~\ref{sec:conclusion}, we summarize our results.

\section{Preliminaries}\label{sec:prelim}
In this section, we recall useful definitions and we introduce the notion of monotonic Markov decision process. We also state the problems that we study. 

\paragraph{Functions and probability distributions.}  For any (partial or total) function $f$, we denote by $\domain(f)$ the domain of definition of $f$. For all sets $A, B$, we denote by $\totfunc(A, B) = \{f : A \rightarrow B \mid \domain(f) = A\}$ the set of total functions from $A$ to $B$. A \textit{probability distribution} over a finite set $A$ is a total function $\pi : A \rightarrow [0,1]$ such that $\sum_{a \in A} \pi(a) = 1$. Its \textit{support} is the set $\support(\pi) = \{a \in A \mid \pi(a) > 0\}$.
We denote by $\distrset(A)$ the set of probability distributions over $A$. 

\paragraph{Stochastic models.} A \textit{discrete-time Markov chain (MC)} is a tuple $(S, \probmat)$ where $S$ is a finite set of states and $\probmat : S \rightarrow \distrset(S)$ is a stochastic transition matrix. For all $s, s' \in S$, we often write $\probmat(s, s')$ for $\probmat(s)(s')$. 
A \textit{path} is an infinite sequence of states $\rho = s_0s_1s_2 \ldots $ such that $\probmat(s_i, s_{i+1}) > 0$ for all $i \geq 0$. Finite paths are defined similarly, and $\probmat$ is naturally extended to finite paths. 

\noindent
A \textit{Markov decision process (MDP)} is a tuple $(S, \ActionsO, \probmat)$ where $S$ is a finite set of states, $\ActionsO$ is a finite set of actions and $\probmat : S \times \ActionsO \rightarrow \distrset(S)$ is a partial stochastic transition function. We often write $\probmat(s, \actionO, s')$ for $\probmat(s, \actionO)(s')$. For each state $s \in S$, we denote by $\enabledactions \subseteq \ActionsO$ the set of enabled actions in $s$, where an action $\actionO \in \ActionsO$ is \textit{enabled} in $s$ if $(s, \actionO) \in \domain(\probmat)$. For all state $s \in S$, we require $\enabledactions \neq \emptyset$, and we thus say that the MDP is $\ActionsO$\textit{-non-blocking}. For all action $\actionO \in \ActionsO$, we also introduce notation $S_\actionO$ for the set of states in which $\actionO$ is enabled.
For $s \in S$ and  $\actionO \in \enabledactions$, we denote by $\suc(s, \actionO) = \support(\probmat(s, \actionO))$ the set of possible successors of $s$ for enabled action $\actionO$. 

\paragraph{Strategies.} Let $(S, \ActionsO, \probmat)$ be an MDP. A \textit{memoryless strategy} is a total function $\lambda : S \rightarrow \ActionsO$ mapping each state $s$ to an enabled action $\actionO \in \enabledactions$. We denote by $\Lambda$ the set of all memoryless strategies. A memoryless strategy $\lambda$ induces an MC $(S, \probmat_\lambda)$ such that for all $s, s' \in S$, $\probmat_\lambda(s, s') = \probmat(s, \lambda(s), s')$.

\paragraph{Costs and value functions.}
Additionally to an MDP $(S, \ActionsO, \probmat)$, we consider a partial \textit{cost function} $\reward : S \times \ActionsO \rightarrow \R$ with $\domain(\reward) = \domain(\probmat)$ that associates a cost with a state $s$ and an enabled action $\actionO$ in $s$. A memoryless strategy $\lambda$ assigns a total cost function $\reward_\lambda : S \rightarrow \R$ to the induced MC $(S, \probmat_\lambda)$, such that $\reward_\lambda(s) = \reward(s, \lambda(s))$. 
Given a path $\rho = s_0s_1s_2 \ldots $ in this MC, 
the \textit{mean-payoff} of $\rho$ is $\MP(\rho) = \limsup_{n \rightarrow \infty} \frac{1}{n}\sum_{i = 0}^{n-1}  \reward_\lambda(s_i)$. Given a subset $G \subseteq S$ of \textit{goal} states and a finite path $\rho$ reaching a state of $G$, the \textit{truncated sum up to $G$} of $\rho$ is $\TS_G(\rho) = \sum_{i = 0}^{n-1}  \reward_\lambda(s_i)$ where $n$ is the first index such that $s_n \in G$.

\noindent
Given an MDP with a cost function $\reward$, and a memoryless strategy $\lambda$, we consider two classical value functions of $\lambda$ defined as follows. 
For all state $s \in S$, 
the \textit{expected mean-payoff} of $\lambda$ is $\EMP_\lambda(s) = \lim_{n \to \infty} \frac{1}{n} \sum_{i=0}^{n-1}\probmat_\lambda^i \reward_\lambda (s)$. Given a subset $G \subseteq S$, and assuming that $\lambda$ reaches $G$ from state $s$ with probability 1, the \textit{expected truncated sum up to $G$} of $\lambda$ is $\ETP_\lambda(s) = \sum_{\rho}\probmat_\lambda(\rho) \TS_G (\rho)$ where the sum is over all finite paths $\rho = s_0 s_1 \ldots s_n$ such that $s_0 = s$, $s_n \in G$, and $s_0, \ldots, s_{n-1} \not \in G$.
Let $\lambda^*$ be a memoryless strategy. Given a value function $\Edot_{\lambda} \in \{\EMP_{\lambda},\ETP_{\lambda}\}$, we say that $\lambda^*$ is \textit{optimal} if $\Edot_{\lambda^*}(s) = \inf_{\lambda \in \Lambda} \Edot_{\lambda}(s)$ for all $s \in S$, and $\Edot_{\lambda^*}$ is called the \textit{optimal} value function.\footnote{\label{fn:altobj}An alternative objective might be to maximize the value function, in which case $\lambda^*$ is optimal if $\Edot_{\lambda^*}(s) = \sup_{\lambda \in \Lambda} \Edot_{\lambda}(s)$ for all $s \in S$.} 
Note that we might have considered other classes of strategies but it is known that for these value functions, there always exists a memoryless strategy that minimizes the expected value of all states~\cite{DBLP:books/daglib/0020348,putermanMDP}. 

\paragraph{Studied problems.} In this paper, we study algorithms for solving MDPs for two quantitative settings: the expected mean-payoff and the stochastic shortest path. Let $(S, \ActionsO, \probmat)$ be an MDP and $\reward : S \times \ActionsO \rightarrow \R$ be a cost function. $(i)$~The \textit{expected mean-payoff (EMP) problem} is to synthesize an optimal strategy for the expected mean-payoff value function. As explained above, such a memoryless optimal strategy always exists, and the problem is solvable in polynomial time via linear programming~\cite{putermanMDP,filar1996competitive}. $(ii)$~When $\reward$ is restricted to \textit{strictly positive} values in $\R_{>0}$, and a subset $G \subseteq S$ of goal states is given, the \textit{stochastic shortest path (SSP) problem} is to synthesize an optimal strategy for the expected truncated sum value function, among the set of strategies that reach $G$ with probability $1$, provided such strategies exist. For all $s \in S$, we denote by $\Lambda_s^P$ the set of \textit{proper strategies} for $s$ that are the strategies that lead from $s$ to $G$ with probability $1$.
Solving the SSP problem consists in two steps. The first step is to determine the set $S^P = \{s \in S \mid \Lambda_s^P \neq \emptyset\}$ of \textit{proper} states, i.e. states having at least one proper strategy. The second step consists in synthesizing an optimal strategy $\lambda^*$ such that $\ETP_{\lambda^*}(s) = \inf_{\lambda \in \Lambda_s^P}\ETP_\lambda(s)$ for all states of $s \in S^P$. It is known that memoryless optimal strategies exist for the SSP, and the problem can be solved in polynomial time through linear programming~\cite{bertsekas1996neuro,bertsekas1991analysis}. Note that the existence of at least one proper strategy for each state is often stated as an assumption on the MDP. It that case, an algorithm for the SSP problem is limited to the second step.

\bigskip
In~\cite{DBLP:conf/qest/WimmerBBHCHDT10}, the authors present a BDD based \textit{symblicit} algorithm for the EMP problem, that is, an algorithm that efficiently combines symbolic and explicit representations. In this paper, we are interested in proposing antichain based (instead of BDD based) symblicit algorithms for both the EMP and SSP problems. Due to the use of antichains, our algorithms apply on a particular, but natural, class of MDPs, called \textit{monotonic} MDPs. We first recall the definition of antichains and related notions. We then consider an example to intuitively illustrate the notion of monotonic MDP and we conclude with its formal definition.

\paragraph{Closed sets and antichains.} Let $S$ be a finite set equipped with a partial order $\preceq$ such that $(S,\preceq)$ is a \textit{semilattice}, i.e. for all $s, s' \in S$, their greatest lower bound $s \sqcap s'$ always exists. A set $L \subseteq S$ is \textit{closed} for $\preceq$ if for all $s \in L$ and all $s' \preceq s$, we have $s' \in L$. If $L_1, L_2 \subseteq S$ are two closed sets, then $L_1\cap L_2$ and $L_1 \cup L_2$ are closed, but $L_1 \diff L_2$ is not necessarily closed. The \textit{closure} $\antclos L$ of a set $L$ is the set $\antclos L = \{s' \in S$ $|$ $\exists s \in L \cdot s' \preceq s\}$. Note that $\antclos L = L$ for all closed sets $L$. 
A set $\alpha$ is an \textit{antichain} if all its elements are pairwise incomparable with respect to $\preceq$. For $L \subseteq S$, we denote by $\lceil L \rceil$ the set of its maximal elements, that is $\lceil L \rceil = \{s \in L \mid \forall s' \in L \cdot s \preceq s' \Rightarrow s = s'\}$. This set $\lceil L \rceil$ is an antichain. If $L$ is closed, then  $\antclos\lceil L \rceil = L$, and $\lceil L \rceil$ is called the \textit{canonical representation} of $L$. The interest of antichains is that they are \textit{compact} representations of closed sets.

\begin{example}\label{ex:monkey}
To illustrate the notion of monotonic MDP in the SSP context, we consider the following example, inspired from~\cite{russell1995artificial}, where a monkey tries to reach an hanging bunch of bananas. There are several items strewn in the room that the monkey can get and use, individually or simultaneously. There is a \textit{box} on which it can climb to get closer to the bananas, a \textit{stone} that can be thrown at the bananas, a \textit{stick} to try to take the bananas down, and obviously the \textit{bananas} that the monkey wants to eventually obtain. Initially, the monkey possesses no item. The monkey can make actions whose effects are to add and/or to remove items from its inventory. We add stochastic aspects to the problem. For example, using the \textit{stick}, the monkey has probability $\frac{1}{5}$ to obtain the \textit{bananas}, while combining the \textit{box} and the \textit{stick} increases this probability to $\frac{1}{2}$. Additionally, we associate a (positive) cost with each action, representing the time spent executing the action. For example, picking up the \textit{stone} has a cost of $1$, while getting the \textit{box} costs $5$. The objective of the monkey is then to minimize the expected cost for reaching the \textit{bananas}.

This kind of specification naturally defines an MDP. The set $S$ of states of the MDP is the set of all the possible combinations of items. Initially the monkey is in the state with no item. The available actions at each state $s \in S$ depend on the items of $s$. For example, when the monkey possesses the \textit{box} and the \textit{stick}, it can decide to try to reach the \textit{bananas} by using one of these two items, or the combination of both of them. If it decides to use the \textit{stick} only, it will reach the state $s \cup \{\textit{bananas}\}$ with probability $\frac{1}{5}$ whereas it will stay at state $s$ with probability $\frac{4}{5}$. This MDP is monotonic in the following sense.
First, the set $S$ is a closed set equipped with the partial order~$\supseteq$. 
Second, the action of trying to reach the \textit{bananas} with the \textit{stick} is also available if the monkey possesses the \textit{stick} together with other items. Moreover, if it succeeds (with probability $\frac{1}{5}$), it will reach a state with the \textit{bananas} and all the items it already had at its disposal. In other words, for all states $s' \in S$ such that $s' \supseteq s = \{\textit{stick}\}$, we have that $\ActionsO_{s} \subseteq \ActionsO_{s'}$, and $t' \supseteq t = \{ \textit{bananas}, \textit{stick}\}$ with $t'$ the state reached from $s'$  with probability $\frac{1}{5}$. Finally, note that the set of goal states $G = \{s \in S \mid \textit{bananas} \in s\}$ is closed. 
\end{example}

\paragraph{New definition of MDPs.} To properly define the notion of monotonic MDPs, we need a slightly different, but equivalent, definition of MDPs which is based on a set $\ActionsI$ of stochastic actions. In this definition, an MDP $M$ is a tuple $(S, \ActionsO, \ActionsI, \edges, \distr)$ where $S$ is a finite set of states, $\ActionsO$ and $\ActionsI$ are two finite sets of actions such that $\ActionsO \cap \ActionsI = \emptyset$, $\edges : S \times \ActionsO \rightarrow \totfunc(\ActionsI, S)$ is a partial successor function, and $\distr : S \times \ActionsO \rightarrow \distrset(\ActionsI)$ is a partial stochastic function such that $\domain(\edges) = \domain(\distr)$. Figure~\ref{fig:exmdp} intuitively illustrates the relationship between the two definitions. 

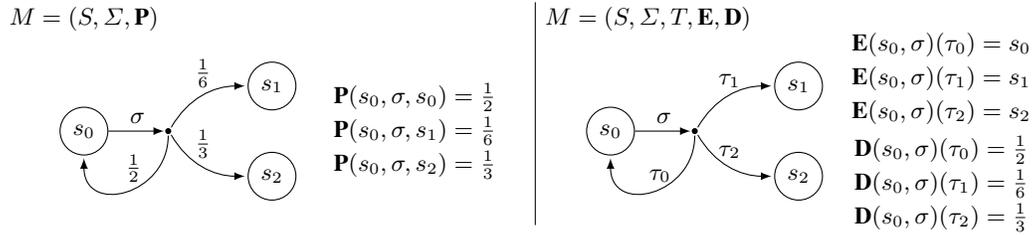
\begin{figure}
\centering
\vspace{-0.4cm}
  \begin{tikzpicture}[>=stealth',shorten >=1pt,auto,node distance=2.5cm,bend angle=0,scale=1,font=\small]
   \tikzstyle{p1}=[draw,circle,text centered,minimum size=3mm]
    \node[p1]  (0)  at (0, 0) {$s_0$};
    \node[p1]  (1)  at (2.5, 0.6) {$s_1$};
    \node[p1]  (2)  at (2.5, -0.6) {$s_2$};
     \fill (1.12,0) circle (0.04cm);

    \node[p1]  (3)  at (7, 0) {$s_0$};
    \node[p1]  (4)  at (9.5, 0.6) {$s_1$};
    \node[p1]  (5)  at (9.5, -0.6) {$s_2$};
     \fill (8.12,0) circle (0.04cm);

    \node at (4.4,0.45) {$\probmat(s_0, \sigma, s_0) = \frac{1}{2}$};
    \node at (4.4,0) {$\probmat(s_0, \sigma, s_1) = \frac{1}{6}$};     
    \node at (4.4,-0.45) {$\probmat(s_0, \sigma, s_2) = \frac{1}{3}$};

    \node at (11.4,1.15) {$\edges(s_0, \sigma)(\actionI_0) = s_0$};
    \node at (11.4,0.7) {$\edges(s_0, \sigma)(\actionI_1) = s_1$};     
    \node at (11.4,0.25) {$\edges(s_0, \sigma)(\actionI_2) = s_2$};
    \node at (11.4,-0.25) {$\distr(s_0, \sigma)(\actionI_0) = \frac{1}{2}$};
    \node at (11.4,-0.7) {$\distr(s_0, \sigma)(\actionI_1) = \frac{1}{6}$};     
    \node at (11.4,-1.15) {$\distr(s_0, \sigma)(\actionI_2) = \frac{1}{3}$};

    \draw (6,-1.25) --  (6,1.75);

    \node at (0,1.5) {$M = (S, \ActionsO, \probmat)$};
    \node at (7.5,1.5) {$M = (S, \ActionsO, \ActionsI, \edges, \distr)$};
    \path
	(0) -- (0);
    \draw[->,>=latex] (0) to[out=0,in=180] node [above, xshift=0mm,yshift=0mm] {$\actionO$} (1.1, 0);
    \draw[->,>=latex] (1.12, -0.06) to[out=270,in=270, distance=0.8cm] node [above, xshift=1mm,yshift=0mm] {$\frac{1}{2}$} (0);
    \draw[->,>=latex] (1.16, 0.05) to[out=60,in=180] node [above, xshift=0mm] {$\frac{1}{6}$} (1);
    \draw[->,>=latex] (1.16, -0.05) to[out=300,in=180] node [above, xshift=0mm,yshift=0mm] {$\frac{1}{3}$} (2);

    \draw[->,>=latex] (3) to[out=0,in=180] node [above, xshift=0mm,yshift=0mm] {$\actionO$} (8.1, 0);
    \draw[->,>=latex] (8.12, -0.06) to[out=270,in=270, distance=0.8cm] node [above, xshift=1mm,yshift=0mm] {$\actionI_0$} (3);
    \draw[->,>=latex] (8.16, 0.05) to[out=60,in=180] node [above, xshift=0mm] {$\actionI_1$} (4);
    \draw[->,>=latex] (8.16, -0.05) to[out=300,in=180] node [above, xshift=0mm,yshift=0mm] {$\actionI_2$} (5);
\end{tikzpicture}
\vspace{-0.3cm}
\caption{Illustration of the new definition of MDPs for a state $s_0 \in S$ and an action $\actionO \in \Sigma_{s_0}$.}
\label{fig:exmdp}
\end{figure}

Let us explain this relationship more precisely. Let an MDP as given in the new definition. We can then derive from $\edges$ and $\distr$ the partial transition function $\probmat : S \times \ActionsO \rightarrow \distrset(S)$ such that for all $s, s' \in S$ and $\actionO \in \enabledactions$, 
\begin{equation*}
\probmat(s,\actionO)(s') = \sum_{\begin{matrix}\actionI \in \ActionsI \\ \edges(s,\actionO)(\actionI)=s'\end{matrix}}\distr(s,\actionO)(\actionI).
\end{equation*}
Conversely, let an MDP $(S, \ActionsO, \probmat)$ as in the first definition. Then we can choose a set $\ActionsI$ of stochastic actions of size $|S|$ and adequate functions $\edges, \distr$ to get the second definition $(S, \ActionsO, \ActionsI, \edges, \distr)$ for this MDP (see Figure~\ref{fig:exmdp}).

In this new definition of MDPs, for all $s \in S$ and all pair of actions $(\actionO, \actionI) \in \ActionsO \times \ActionsI$, there is at most one $s' \in S$ such that $\edges(s, \actionO)(\actionI) = s'$. We thus say that $M$ is \textit{deterministic}. Moreover, since for all pair $(s, \actionO) \in \domain(\edges)$, $\edges(s, \actionO)$ is a total function mapping each $\actionI \in \ActionsI$ to a state $s \in S$, we say that $M$ is $\ActionsI$\textit{-complete}.

Notice that the notion of MC induced by a strategy can also be described in this new formalism as follows. Given an MDP $(S, \ActionsO, \ActionsI, \edges, \distr)$ and a memoryless strategy $\lambda$, we have the induced MC $(S, \ActionsI, \edges_{\lambda}, \distr_{\lambda})$ such that $\edges_\lambda : S \rightarrow \distrset(\ActionsI)$ is the successor function with $\edges_\lambda(s) = \edges(s, \lambda(s))$, for all $s \in S$, and $\distr_\lambda : S \rightarrow \distrset(\ActionsI)$ is the stochastic function with $\distr_\lambda(s) = \distr(s, \lambda(s))$, for all $s \in S$.

Depending on the context, we will use both definitions $M = (S, \ActionsO, \ActionsI, \edges, \distr)$ and $M = (S, \ActionsO, \probmat)$ for MDPs, assuming that $\probmat$ is always obtained from some set $\ActionsI$ and partial functions $\edges$ and $\distr$.
We can now formally define the notion of monotonic MDP.

\paragraph{Monotonic MDPs.} A \textit{monotonic MDP} is an MDP $\monMDP = (S, \ActionsO, \ActionsI, \edges, \distr)$ such that:
\begin{enumerate} 
	\item The set $S$ is equipped with a partial order $\preceq$ such that $(S,\preceq)$ is a semilattice. 
	\item The partial order $\preceq$ is \textit{compatible} with $\edges$, i.e. for all $s, s' \in S$, if $s \preceq s'$, then  for all $\actionO \in \ActionsO$, $\actionI \in \ActionsI$, for all $t' \in S$ such that $\edges(s', \actionO)(\actionI) = t'$, there exists $t \in S$ such that $\edges(s, \actionO)(\actionI) = t$ and $t \preceq t'$.
\end{enumerate}

\noindent Note that since $(S,\preceq)$ is a semilattice, we have that $S$ is closed for $\preceq$. With this definition, and in particular by compatibility of $\preceq$, we have the next proposition.

\begin{proposition}\label{prop:As}
The following statements hold for a monotonic MDP $\monMDP$:
\begin{itemize}
\item For all $s,s' \in S$, if $s \preceq s'$ then $\enabledactionsprime \subseteq \enabledactions$
\item For all $\actionO \in \ActionsO$, $\enabledstates$ is closed. 
\end{itemize}
\end{proposition}

\begin{remark}
In this definition, by monotonic MDPs, we mean MDPs that are built on state spaces \textit{already equipped with a natural partial order}. For instance, this is the case for the two classes of MDPs studied in Section~\ref{sec:experiments}. 
The same kind of approach has already been proposed in~\cite{DBLP:journals/tcs/FinkelS01}.

Note that all MDPs can be seen monotonic. Indeed, let $(S, \Sigma, T, \edges, \distr)$ be a given MDP and let $\preceq$ be a partial order such that all states in $S$ are pairwise incomparable with respect to $\preceq$. Let~$t\not\in S$ be an additional state such that $(1)$ $t \preceq s$ for all $s \in S$, $(2)$ $\edges(s, \sigma)(\tau) \neq t$ for all $s \in S, \sigma \in \Sigma, \tau \in T$, and $(3)$~$\edges(t, \sigma)(\tau) = t$ for all $\sigma \in \Sigma, \tau \in T$. Then, we have that $(S\cup\{t\}, \preceq)$ is a semilattice and $\preceq$ is compatible with $\edges$. However, such a partial order would not lead to efficient algorithms in the sense studied in this paper. 
\end{remark}

\section{Strategy iteration algorithms}\label{sec:si}
In this section, we present strategy iteration algorithms for synthesizing optimal strategies for the SSP and EMP problems. A \textit{strategy iteration} algorithm~\cite{howard1960dynamic} consists in generating a sequence of monotonically improving strategies (along with their associated value functions) until converging to an optimal one. Each iteration is composed of two phases: the \textit{strategy evaluation phase} in which the value function of the current strategy is computed, and the \textit{strategy improvement phase} in which the strategy is improved (if possible) at each state, by using the preceding computed value function. The algorithm stops after a finite number of iterations, as soon as no more improvement can be made, and returns the computed optimal strategy. 

We now describe two strategy iteration algorithms, for the SSP and the EMP. We follow the presentation of those algorithms as given in \cite{DBLP:conf/qest/WimmerBBHCHDT10}.

\subsection{Stochastic shortest path}\label{subsec:SSP}
We start with the strategy iteration algorithm for the SSP problem~\cite{howard1960dynamic,bertsekas1996neuro}. Let $M = (S, \ActionsO, \probmat)$ be an MDP, $\reward : S \times \ActionsO \rightarrow \R_{>0}$ be a strictly positive cost function, and $G \subseteq S$ be a set of goal states. Recall from the previous section that the solution to the SSP problem is to first compute the set of proper states which are the states having at least one proper strategy. 

\paragraph{Computing proper states.} An algorithm is proposed in~\cite{DBLP:conf/concur/Alfaro99} for computing in quadratic time  the set $S^P = \{s \in S \mid \Lambda_s^P \neq \emptyset\}$ of proper states. To present it, given two subsets $X, Y \subseteq S$, we define the predicate $\apre(Y, X)$ such that for all $s \in S$, 
$$ s \models \apre(Y, X) \Leftrightarrow \exists \actionO \in \ActionsO_s, (\suc(s, \actionO) \subseteq Y \land \suc(s,\actionO) \cap X \neq \emptyset).$$
Then, we can compute the set $S^P$ of proper states by the following $\mu$-calculus expression:
$$ S^P = \nu Y \cdot \mu X \cdot (\apre(Y,X) \lor \mathsf{G}),$$
where we denote by $\mathsf{G}$ a predicate that holds exactly for the states in $G$. The algorithm works as follows. Initially, we have $Y_0 = S$. At the end of the first iteration, we have $Y_1 = S \diff C_0$, where $C_0$ is the set of states that reach $G$ with probability $0$. At the end of the second iteration, we have $Y_2 = Y_1 \diff C_1$, where $C_1$ is the set of states that cannot reach $G$ without risking to enter $C_0$ (i.e. states in $C_1$ have a strictly positive probability of entering $C_0$). More generally, at the end of iteration $k > 0$, we have $Y_k = Y_{k-1} \diff C_{k-1}$, where $C_{k-1}$ is the set of states that cannot reach $G$ without risking to enter $\bigcup_{i=0}^{k-2}C_i$. The correctness and complexity results are proved in~\cite{DBLP:conf/concur/Alfaro99}.

Given an MDP $M = (S, \ActionsO, \probmat)$ with a cost function $\reward$ and a set $G \subseteq S$, one can restrict $M$ and $\reward$ to the set $S^P$ of proper states. We obtain a new MDP $M^P = (S^P, \Sigma, \probmat^P)$ with cost function $\reward^P$ such that $\probmat^P$ and $\reward^P$ are the restriction of $\probmat$ and $\reward$ to $S^P$. Moreover, for all state $s \in S^P$, we let $\ActionsO_s^P = \{\actionO \in \enabledactions \mid \suc(s, \actionO) \subseteq S^P\}$ be the set of enabled actions in $s$. Note that by construction of $S^P$, we have $\ActionsO_s^P \neq \emptyset$ for all $s \in S^P$, showing that $M^P$ is  $\ActionsO$\textit-non-blocking. To avoid a change of notation, in the sequel of this subsection, we make the assumption that each state of $M$ is proper.

\paragraph{Strategy iteration algorithm.}
The strategy iteration algorithm for SSP, named \textsc{SSP\_StrategyIteration}, is given in Algorithm~\ref{algo:SSPSI}\footnote{If the expected truncated sum has to be maximized, the cost function is restricted to the strictly negative real numbers and $\arg\min$ is replaced by $\arg\max$ in line 4.}. This algorithm is applied under the typical assumption that all cycles in the underlying graph of $M$ have strictly positive cost~\cite{bertsekas1996neuro}. This assumption holds in our case by definition of the cost function $\reward$.
The algorithm starts with an arbitrary proper strategy $\lambda_0$, that can be easily computed with the algorithm of~\cite{DBLP:conf/concur/Alfaro99}, and improves it until an optimal strategy is found. The expected truncated sum $v_n$ of the current strategy $\lambda_n$ is computed by solving the system of linear equations in line 3, and used to improve the strategy (if possible) at each state. Note that the strategy $\lambda_n$ is improved at a state $s$ to an action $\sigma \in \enabledactions$ only if the new expected truncated sum is strictly smaller than the expected truncated sum of the action $\lambda_n(s)$, i.e. only if $\lambda_n(s) \not\in \underset{\sigma \in \enabledactions}{\arg\min}(\reward(s,\actionO) + \underset{s'\in S}{\sum}\probmat(s, \actionO, s')\cdot v_n(s'))$. If no improvement is possible for any state, an optimal strategy is found and the algorithm terminates in line 7. Otherwise, it restarts by solving the new equation system, tries to improve the strategy using the new values computed, and so on.

\begin{algorithm}[t]
\caption{\textsc{SSP\_StrategyIteration}$($MDP $M$, Strictly positive cost function $\reward$, Goal states $G)$}
\begin{algorithmic}[1] \label{algo:SSPSI}
\STATE $n := 0, \lambda_n :=$ \textsc{InitialProperStrategy}$(M, G)$
\REPEAT
	\STATE Obtain $v_n$ by solving 
	\begin{equation*} 
	\hspace{-9.8cm}	 \reward_{\lambda_n} + (\probmat_{\lambda_n} - \id)v_n = 0
 	\end{equation*}
	\STATE $\widehat{\enabledactions} := \underset{\actionO \in \ActionsO_s}{\arg \min} (\reward(s,\actionO) + \underset{s'\in S}{\sum}\probmat(s, \actionO, s')\cdot v_n(s')), \forall s \in S$
	\STATE Choose $\lambda_{n+1}$ such that $\lambda_{n+1}(s) \in \widehat{\enabledactions}, \forall s \in S$, setting $\lambda_{n+1}(s) := \lambda_n(s)$ if possible.
	\STATE $n := n+1$
\UNTIL{$\lambda_{n} = \lambda_{n-1}$}
\STATE return $(\lambda_{n-1}, v_{n-1})$
\end{algorithmic}
\end{algorithm}

\subsection{Expected mean-payoff}
We now consider the strategy iteration algorithm for the EMP problem~\cite{veinott1966finding,putermanMDP} (see Algorithm~\ref{algo:HV}\footnote{\label{fn:MPmax}If the expected mean-payoff has to be maximized, one has to replace $\arg\min$ by $\arg\max$ in lines 4 and 7.}). More details can be found in \cite{putermanMDP}. The algorithm starts with an arbitrary strategy $\lambda_0$ (here any initial strategy is appropriate). By solving the equation system of line 3, we obtain the gain value $g_n$ and bias value $b_n$ of the current strategy $\lambda_n$. The gain corresponds to the expected mean-payoff, while the bias can be interpreted as the expected total difference between the cost and the expected mean-payoff. 
The computed gain value is then used to locally improve the strategy (lines 4-5). If such an improvement is not possible for any state, the bias value is used to locally improve the strategy (lines 6-7). By improving the strategy with the bias value, only actions that also optimize the gain can be considered (see set $\widehat{\enabledactions}$). 
Finally, the algorithm stops at line $10$ as soon as none of those improvements can be made for any state, and returns the optimal strategy $\lambda_{n-1}$ along with its associated expected mean-payoff. 

\begin{algorithm}[t]
\caption{\textsc{EMP\_StrategyIteration}$($MDP $M$, Cost function $\reward)$}
\begin{algorithmic}[1] \label{algo:HV}
\STATE $n := 0, \lambda_n :=$ \textsc{InitialStrategy}$(M)$
\REPEAT
	\STATE Obtain $g_n$ and $b_n$ by solving 

	\begin{equation*} 
		\hspace{-8.8cm}\begin{cases}
			\ (\probmat_{\lambda_n} - \id)g_n = 0 \\
			\ \reward_{\lambda_n} - g_n + (\probmat_{\lambda_n} - \id)b_n = 0\\
			\ \probmat_{\lambda_n}^*b_n = 0
		\end{cases} 
 	\end{equation*}

	\STATE $\widehat{\enabledactions} := \underset{\actionO \in \ActionsO_s}{\arg \min} \underset{s'\in S}{\sum} \probmat(s, \actionO, s')\cdot g_n(s'), \forall s \in S$
	\STATE Choose $\lambda_{n+1}$ such that $\lambda_{n+1}(s) \in \widehat{\enabledactions}, \forall s \in S$, setting $\lambda_{n+1}(s) := \lambda_n(s)$ if possible.
	\IF{$\lambda_{n+1} = \lambda_n$}
		\STATE Choose $\lambda_{n+1}$ such that $\lambda_{n+1}(s) \in \underset{\sigma \in \widehat{\enabledactions}}{\arg \min} (\reward(s,\actionO) + \underset{s'\in S}{\sum}\probmat(s, \actionO, s')\cdot b_n(s')), \forall s \in S$, \\setting $\lambda_{n+1}(s) = \lambda_n(s)$ if possible.
	\ENDIF
	\STATE $n := n+1$
\UNTIL{$\lambda_{n} = \lambda_{n-1}$}
\STATE return $(\lambda_{n-1}, g_{n-1})$
\end{algorithmic}
\end{algorithm}

\section{Symblicit approach} \label{sec:symblicit}

Explicit-state representations of MDPs like sparse-matrices are often limited to the available memory. When treating MDPs with large state spaces, using explicit representations quickly becomes unfeasible. Moreover, the linear systems of large MDPs are in general hard to solve. Symbolic representations with \textit{(Multi-terminal) Binary Decision Diagrams ((MT)BDDs)} are then an alternative solution. A \textit{BDD}~\cite{DBLP:journals/tc/Bryant86} is a data structure that permits to compactly represent boolean functions of $n$ boolean variables, i.e. $\{0,1\}^n \rightarrow \{0,1\}$. An \textit{MTBDD}~\cite{DBLP:journals/fmsd/FujitaMY97} is a generalization of a BDD used to represent functions of $n$ boolean variables, i.e. $\{0,1\}^n \rightarrow V$, where $V$ is a finite set. A symblicit algorithm for the EMP problem has been studied in~\cite{DBLP:conf/qest/WimmerBBHCHDT10}. It combines symbolic techniques based on (MT)BDDs with explicit representations and often leads to a good trade-off between execution time and memory consumption. 

In this section, we recall the symblicit algorithm proposed in~\cite{DBLP:conf/qest/WimmerBBHCHDT10} for solving the EMP problem on MDPs. However, our description is more general to suit also for the SSP problem. We first talk about bisimulation lumping, a technique used by this symblicit algorithm to reduce the state space of the models it works on.

\subsection{Bisimulation lumping}
The \textit{bisimulation lumping} technique~\cite{kemeny1960jl,DBLP:journals/iandc/LarsenS91,buchholz1994exact} applies to Markov chains. It consists in gathering certain states of an MC which behave equivalently according to the class of properties under consideration. For the expected truncated sum and the expected mean-payoff, the following definition of equivalence of two states can be used. Let $(S, \probmat)$ be an MC and $\reward : S \rightarrow \R$ be a cost function on $S$. 
Let $\sim$ be an equivalence relation on $S$ and $S_{\sim}$ be the induced partition. We call \textit{block} of $S_{\sim}$ any equivalence class of $\sim$.
We say that $\sim$ is a \textit{bisimulation} if for all $s,t \in S$ such that $s \sim t$, we have $\reward(s) = \reward(t)$ and $\probmat(s,C) = \probmat(t,C)$ for all block $C \in S_{\sim}$, where $\probmat(s,C) = \sum_{s'\in C}\probmat(s,s')$. 

Let $(S, \probmat)$ be an MC with cost function $\reward$, and $\sim$ be a bisimulation on $S$. The \textit{bisimulation quotient} is the MC $(S_{\sim}, \probmat_{\sim})$ such that $\probmat_{\sim}(C,C') = \probmat(s,C')$, where $s \in C$ and $C,C' \in S_{\sim}$. The cost function $\reward_{\sim} : S_{\sim} \rightarrow \R$ is transferred to the quotient such that $\reward_{\sim}(C) = \reward(s)$, where $s \in C$ and $C \in S_{\sim}$. The quotient is thus a minimized model equivalent to the original one for our purpose, since it satisfies properties like expected truncated sum and expected mean-payoff as the original model~\cite{DBLP:journals/iandc/BaierKHW05}. Usually, we are interested in the unique \textit{largest} bisimulation, denoted $\equivlump$, which leads to the smallest bisimulation quotient $(S_{\equivlump}, \probmat_{\equivlump})$.

Algorithm \textsc{Lump}~\cite{DBLP:journals/ipl/DerisaviHS03} (see Algorithm~\ref{algo:lump}) describes how to compute the partition induced by the largest bisimulation. This algorithm is based on Paige and Tarjan's algorithm for computing bisimilarity of labeled transition systems~\cite{DBLP:journals/siamcomp/PaigeT87}.

\begin{algorithm}[t]
\caption{\textsc{Lump}$($MC $M$, Cost function $\reward)$}
\begin{algorithmic}[1] \label{algo:lump}
\STATE $P := \textsc{InitialPartition}(M, \reward)$
\STATE $L := P$
\WHILE{$L \neq \emptyset$}
	\STATE $C := \textsc{Pop}(L)$
	\STATE $P_\textnormal{new} := \emptyset$
	\FORALL{$B \in P$}
		\STATE $\{B_1,\dots,B_k\} := \textsc{SplitBlock}(B, C)$
		\STATE $P_\textnormal{new} := P_\textnormal{new} \cup \{B_1,\dots,B_k\}$
		\STATE $B_l :=$ some block in $\{B_1,\dots,B_k\}$
		\STATE  $L := L \cup \{B_1,\dots,B_k\}\diff B_l$
	\ENDFOR
	\STATE $P := P_{\textnormal{new}}$
\ENDWHILE
\STATE return $P$
\end{algorithmic}
\end{algorithm}

For a given MC $M = (S, \probmat)$ with cost function $\reward$, Algorithm \textsc{Lump} first computes the initial partition $P$ such that for all $s, t \in S$, $s$ and $t$ belongs to the same block of $P$ iff \reward$(s)$ = \reward$(t)$. The algorithm holds a list $L$ of potential splitters of $P$, where a \textit{splitter} of $P$ is a set $C \subseteq S$ such that $\exists B \in P, \exists s,s' \in B$ such that \probmat$(s, C)$ $\neq$ \probmat$(s',C)$. Initially, this list $L$ contains the blocks of the initial partition $P$. Then, while $L$ is non empty, the algorithm takes a splitter $C$ from $L$ and refines each block of the partition according to $C$. Algorithm \textsc{SplitBlock} splits a block $B$ into non empty sub-blocks $B_1,\cdots,B_k$ according to the probability of reaching the splitter $C$, i.e. for all $s,s' \in B$, we have $s,s' \in B_i$ for some $i$ iff $\probmat(s,C) = \probmat(s',C)$. The block $B$ is then replaced in $P$ by the computed sub-blocks $B_1, \cdots, B_k$. 
Finally, we add to $L$ the sub-blocks $B_1,\cdots,B_k$, but one which can be omitted since its power of splitting other blocks is maintained by the remaining sub-blocks~\cite{DBLP:journals/ipl/DerisaviHS03}. In general, we prefer to omit the largest sub-block since it might be the most costly to process as potential splitter. 
The algorithm terminates when the list $L$ is empty, which means that the partition is refined w.r.t. all potential splitters, i.e. $P$ is the partition induced by the largest bisimulation $\equivlump$.

\subsection{Symblicit algorithm} 
The algorithmic basis of the symblicit approach is the strategy iteration algorithm (see Algorithm~\ref{algo:SSPSI} for the SSP and Algorithm~\ref{algo:HV} for the EMP). In addition, once a strategy $\lambda_n$ is fixed for the MDP, Algorithm \textsc{Lump} is applied on the induced MC in order to reduce its size and to produce its bisimulation quotient. The system of linear equations is then solved for the quotient, and the computed value functions are used to improve the strategy for each individual state of the MDP. 

The symblicit algorithm is described in Algorithm~\textsc{Symblicit} (see Algorithm~\ref{algo:symblicit}). Note that in line $1$, the initial strategy $\lambda_0$ is selected arbitrarily for the EMP, while it has to be a proper strategy in case of SSP. 
It combines symbolic\footnote{We use calligraphic style for symbols denoting a symbolic representation.} and explicit representations of data manipulated by the underlying algorithm as follows. The MDP $\symbMDP$, the cost function $\symbReward$, the strategies $\lambda_n$, the induced MCs $\symbMC$ with cost functions $\symbRewardMC$, and the set $\symbGoal$ of goal states for the SSP, are symbolically represented. Therefore, the lumping procedure is applied on symbolic MCs and produces a symbolic representation of the bisimulation quotient $\symbQuotient$ and associated cost function $\symbRewardQuotient$ (line 4). However, since solving linear systems is more efficient using an explicit representation of the transition matrix, the computed bisimulation quotient is converted to a sparse matrix representation (line 5). The quotient being in general much smaller than the original model, there is no memory issues by storing it explicitly. The linear system is thus solved on the explicit quotient. The computed value functions $x_n$ (corresponding to $v_n$ for the SPP, and $g_n$ and $b_n$ for the EMP) are then converted into symbolic representations $\symbValue_n$, and transferred back to the original MDP (line 7). Finally, the update of the strategy is performed symbolically.

\begin{algorithm}[t]
\caption{\textsc{Symblicit}$(\textnormal{MDP} \symbMDP, [\textnormal{Strictly positive}] \textnormal{ cost function} \ \symbReward[, \textnormal{Goal states } \symbGoal])$}
\begin{algorithmic}[1] \label{algo:symblicit}
\STATE $n := 0, \lambda_n :=$ \textsc{InitialStrategy}$(\symbMDP[, \symbGoal])$
\REPEAT
	\STATE $(\symbMC, \symbRewardMC) :=$ \textsc{InducedMCAndCost}$(\symbMDP, \symbReward, \lambda_n)$
	\STATE $(\symbQuotient, \symbRewardQuotient) :=$ \textsc{Lump}$(\symbMC, \symbRewardMC)$
	\STATE $(\explQuotient, \explRewardQuotient) :=$ \textsc{Explicit}$(\symbQuotient, \symbRewardQuotient)$
	\STATE $x_n :=$ \textsc{SolveLinearSystem}$(\explQuotient, \explRewardQuotient)$			
	\STATE $\symbValue_n :=$ \textsc{Symbolic}$(x_n)$
	\STATE $\lambda_{n+1} :=$ \textsc{ImproveStrategy}$(\symbMDP, \lambda_n, \symbValue_n)$	
	\STATE $n := n+1$
\UNTIL{$\lambda_{n} = \lambda_{n-1}$}
\STATE return $(\lambda_{n-1}, \symbValue_{n-1})$
\end{algorithmic}
\end{algorithm}

\medskip
In~\cite{DBLP:conf/qest/WimmerBBHCHDT10}, the intermediate symbolic representations use (MT)BDDs. In the sequel, we introduce a new data structure extended from antichains, called \textit{pseudo-antichains}, and we show how it can be used (instead of (MT)BBDs) to solve the SSP and EMP problems for monotonic MDPs under well-chosen assumptions.

\section{Pseudo-antichains}\label{sec:pa}

In this section, we introduce the notion of pseudo-antichains. We start by recalling properties on antichains.
Let $(S, \preceq)$ be a semilattice. We have the next classical properties on antichains \cite{DBLP:journals/fmsd/FiliotJR11}:

\begin{proposition} \label{prop:opantichains}
Let $\alpha_1, \alpha_2 \subseteq S$ be two antichains and $s \in S$. Then:
\begin{itemize}
\item $s \in\  \antclos \alpha_1$ iff $\exists a \in \alpha_1 \cdot s \preceq a$
\item  $\antclos\alpha_1$ $\cup$ $\antclos\alpha_2 =\ \antclos\lceil \alpha_1 \cup \alpha_2\rceil$
\item  $\antclos\alpha_1$ $\cap$ $\antclos\alpha_2 =\ \antclos\lceil \alpha_1 \sqcap \alpha_2\rceil$, 
where $\alpha_1 \sqcap \alpha_2 \defeq \{a_1 \sqcap a_2 \mid a_1 \in \alpha_1,  a_2 \in \alpha_2 \}$
\item $\antclos\alpha_1 \subseteq\ \antclos\alpha_2$ iff $\forall a_1 \in \alpha_1 \cdot \exists a_2 \in \alpha_2 \cdot a_1 \preceq a_2$
\end{itemize}
\end{proposition}

\noindent
For convenience, when $\alpha_1$ and $\alpha_2$ are antichains, we use notation $\alpha_1$ $\antunion$ $\alpha_2$ (resp. $\alpha_1$ $\antinter$ $\alpha_2$) for the antichain $\lceil\antclos\alpha_1$ $\cup \antclos\alpha_2\rceil$ (resp. $\lceil\antclos\alpha_1$ $\cap \antclos\alpha_2\rceil$).

Let $L_1, L_2 \subseteq S$ be two closed sets. Unlike the union or intersection, the difference $L_1\diff L_2$ is not necessarily a closed set. There is thus a need for a new structure that ``represents" $L_1 \diff L_2$ in a compact way, as antichains compactly represent closed sets. In this aim, in the next two sections, we begin by introducing the notion of pseudo-element, and we then introduce the notion of pseudo-antichain. We also describe some properties that can be used in algorithms using pseudo-antichains.

\subsection{Pseudo-elements and pseudo-closures} 
A \textit{pseudo-element} is a couple $(x, \alpha)$ where $x \in S$ and $\alpha \subseteq S$ is an antichain such that $x\not\in$ $\antclos\alpha$. The \textit{pseudo-closure} of a pseudo-element $(x, \alpha)$, denoted by $\pseudclos(x, \alpha)$, is the set $\pseudclos(x, \alpha) = \{s \in S \mid s \preceq x$ and $s \not\in$ $\antclos\alpha\} =$ $\antclos\{x\}\diff\!\!\antclos\alpha$. Notice that $\pseudclos(x, \alpha)$ is non empty since $x\not\in$ $\antclos\alpha$ by definition of a pseudo-element.
The following example illustrates the notion of pseudo-closure of pseudo-elements.

\begin{example} \label{ex:pseudo-element}
Let $\nat^2_{\leq3}$ be the set of pairs of natural numbers in $[0, 3]$ and let $\preceq$ be a partial order on $\nat^2_{\leq3}$ such that $(n_1, n_1') \preceq (n_2, n_2')$ iff $n_1 \leq n_2$ and $n_1' \leq n_2'$. Then, $(\nat^2_{\leq3}, \preceq)$ is a complete lattice with least upper bound $\sqcup$ such that $(n_1, n_1') \sqcup (n_2, n_2') = (\max(n_1, n_2), \max(n_1', n_2'))$, and greatest lower bound $\sqcap$ such that $(n_1, n_1') \sqcap (n_2, n_2') = (\min(n_1, n_2), \min(n_1', n_2'))$.
With $x = (3, 2)$ and $\alpha = \{(2, 1), (0, 2)\}$, the pseudo-closure of the pseudo-element $(x, \alpha)$ is the set $\pseudclos(x, \alpha) = \{(3, 2), (3, 1), (3, 0), (2, 2), (1, 2)\} =\ \antclos\{x\}\setminus\!\!\antclos\alpha$ (see Figure~\ref{fig:pa}).
\end{example}

\begin{figure}[t]
\centering
\begin{minipage}{0.47\linewidth}
     \centering
          \begin{tikzpicture}[-,>=stealth',auto,node distance=2.5cm,bend angle=45,scale=0.68,font=\footnotesize]
	     \fill[light-gray] (-3.7,3) --  (-1, 5.7) -- (2.7, 2) -- (2, 2.7) -- (1, 1.7) -- (-1, 3.7) -- (-2.7, 2);
	    \node (0)  at (0, 0) {$(0,0)$};
	    \node (1)  at (-1,1) {$(1, 0)$};
	    \node (2)  at (1,1) {$(0, 1)$};
	    \node (3)  at (0,2) {$(1, 1)$};
	    \node (4)  at (-2,2) {$(2, 0)$};
	    \node (5)  at (2,2) {$(0, 2)$};
	    \node (6)  at (-3, 3) {$(3, 0)$};
	    \node (7)  at (-1,3) {$(2, 1)$};
	    \node (8)  at (1,3) {$(1, 2)$};
	    \node (9)  at (3,3) {$(0, 3)$};
	    \node (10)  at (-2,4) {$(3, 1)$};
	    \node (11)  at (0,4) {$(2, 2)$};
	    \node (12)  at (2,4) {$(1, 3)$};
	    \node (13)  at (-1,5) {$(3, 2)$};
	    \node (14)  at (1,5) {$(2, 3)$};
	    \node (15)  at (0,6) {$(3, 3)$};
	   \path
		(0) edge[dashed] (1)
	         (0) edge[dashed] (2)
		(1) edge[dashed] (3)
		(1) edge[dashed] (4)
		(2) edge[dashed] (3)
		(2) edge[dashed] (5)
		(3) edge[dashed] (7)
		(3) edge[dashed] (8)
		(4) edge[dashed] (6)
		(4) edge[dashed] (7)
		(5) edge[dashed] (8)
		(5) edge[dashed] (9)
		(6) edge[dashed] (10)
		(7) edge[dashed] (10)
		(7) edge[dashed] (11)
		(8) edge[dashed] (11)
		(8) edge[dashed] (12)
		(9) edge[dashed] (12)
		(10) edge[dashed] (13)
		(11) edge[dashed] (13)
		(11) edge[dashed] (14)
		(12) edge[dashed] (14)
		(13) edge[dashed] (15)
		(14) edge[dashed] (15)
	
		(-3.7,3) edge (-1, 5.7)
		(2.7,2) edge (-1, 5.7)
	
		(-2.7, 2) edge (-1, 3.7)
		(1, 1.7) edge (-1, 3.7)
		(2, 2.7) edge (1, 1.7)
		(2.7, 2) edge (2, 2.7);
	\end{tikzpicture}
	\vspace*{-0.2cm}
	\caption{Pseudo-closure of a pseudo-element over $(\nat^2_{\leq3}, \preceq)$.}
	\vspace*{-0.4cm}	
	\label{fig:pa}
   \end{minipage}
   \hfill
   \begin{minipage}{0.47\linewidth}
	\centering
               \begin{tikzpicture}[-,>=stealth',auto,node distance=2.5cm,bend angle=45,scale=0.83,font=\footnotesize]
		     \fill[light-gray2] (-2.5,-0.5) -- (-0.5,1.5) -- (2,-1) -- (1, -2) -- (0.5,-1.5) -- (0,-2) -- (-1,-1) -- (-1.5,-1.5);
		     \fill[light-gray] (-2,-1) -- (-0.5,0.5) -- (1.5,-1.5) -- (1, -2) -- (-0.5,-0.5) -- (-1.5,-1.5);
		    \node (0)  at (-0.5, 1.8) {$y$};
		    \node (0)  at (0.5, -1.2) {$b_2$};
		    \node (0)  at (-0.5, 0.8) {$x$};
		    \node (0)  at (-1, -0.7) {$b_1$};
		    \node (0)  at (-0.5, -0.2) {$a$};
		    \node (0)  at (-0.5, 1.5) {$\bullet$};
		    \node (0)  at (-0.5, 0.5) {$\bullet$};
		    \node (0)  at (0.5, -1.5) {$\bullet$};    
		    \node (0)  at (-1, -1) {$\bullet$};    
		    \node (0)  at (-0.5, -0.5) {$\bullet$};
		
		    \path
			(-0.5,1.5) edge (-2.5,-0.5)
			(-0.5,1.5) edge (2,-1)
			(-0.5,0.5) edge (-2,-1)
			(-0.5,-0.5) edge (-1.5,-1.5)
			(-0.5,0.5) edge (1.5,-1.5)
			(-0.5,-0.5) edge (1,-2)
			(-1,-1) edge (0,-2)
			(0.5,-1.5) edge (0,-2)
			
			(-2.5,-0.5) edge[ultra thin,dashed] (0.5,2.5)
			(-2,-1) edge[ultra thin,dashed] (1,2)
			(-1.5,-1.5) edge[ultra thin,dashed] (1.5,1.5)
			(-1,-2) edge[ultra thin,dashed] (2,1)
			(-0.5,-2.5) edge[ultra thin,dashed] (2.5,0.5)
			(0,-3) edge[ultra thin,dashed] (3, 0)
			(0,-3) edge[ultra thin,dashed] (-3,0)
			(0.5,-2.5) edge[ultra thin,dashed] (-2.5,0.5)
			(1,-2) edge[ultra thin,dashed] (-2,1)
			(1.5,-1.5) edge[ultra thin,dashed] (-1.5,1.5)
			(2,-1) edge[ultra thin,dashed] (-1,2)
			(2.5,-0.5) edge[ultra thin,dashed] (-0.5,2.5);
	\end{tikzpicture}
	\vspace*{-0.2cm}
	\caption[Inclusion of pseudo-closures of pseudo-elements]{Inclusion of pseudo-closures of pseudo-elements.}
	\vspace*{-0.4cm}
	\label{fig:pa_inclusion}
     \end{minipage}
\end{figure}
 
There may exist two pseudo-elements $(x, \alpha)$ and $(y, \beta)$ such that  $\pseudclos(x, \alpha) =$  $\pseudclos(y, \beta)$. We say that the pseudo-element $(x, \alpha)$ is in \textit{canonical form} if $\forall a \in \alpha \cdot a \preceq x$. The next proposition and its corollary show that the canonical form is unique. Notice that for all pseudo-element $(x, \alpha)$, there exists a pseudo-element in canonical form $(y, \beta)$ such that $\pseudclos(x, \alpha) =$ $\pseudclos(y, \beta)$: it is equal to $(x, \{x\}\ \antinter\ \alpha)$. We say that such a couple $(y, \beta)$ is the \textit{canonical representation} of $\pseudclos(x, \alpha)$.

\begin{proposition} \label{prop:inclusion}
Let $(x, \alpha)$ and $(y, \beta)$ be two pseudo-elements. Then $\pseudclos(x, \alpha) \subseteq$ $\pseudclos(y, \beta)$ iff $x \preceq y$ and $\forall b \in \beta \cdot b \sqcap x \in$ $\antclos\alpha$.
\end{proposition}

\begin{proof}
We prove the two implications:
\begin{itemize}
\item \underline{$\Rightarrow$:} Suppose that $\pseudclos(x, \alpha) \subseteq\ \pseudclos(y, \beta)$ and let us prove that $x \preceq y$ and $\forall b \in \beta \cdot b \sqcap x \in$ $\antclos\alpha$.
As $x \in$ $\pseudclos(x, \alpha) \subseteq$ $\pseudclos (y, \beta)$, then $x \preceq y$. Consider $s = b \sqcap x$ for some $b \in \beta$. We have $s \not\in$ $\pseudclos(y, \beta)$ because $s \preceq b$ and thus $s \not\in$ $\pseudclos(x, \alpha)$. As $s \preceq x$, it follows that $s \in$ $\antclos \alpha$.

\item \underline{$\Leftarrow$:} Suppose that $x \preceq y$ and $\forall b \in \beta \cdot b \sqcap x \in$ $\antclos\alpha$. Let us prove that $\forall s \in$ $\pseudclos(x, \alpha)$, we have $s \in$ $\pseudclos(y, \beta)$. As $s \preceq x$, and $x \preceq y$ by hypothesis, we have $s \preceq y$. Suppose that $s \in$ $\antclos \beta$, that is $s \preceq b$, for some $b \in \beta$. As $s \preceq x$, we have $s \preceq b \sqcap x$ and thus $s \in$ $\antclos \alpha$ by hypothesis. This is impossible since $s \in$ $\pseudclos(x, \alpha)$. Therefore, $s \not\in$ $\antclos\beta$, and thus $s \in $ $\pseudclos(y, \beta)$.
\end{itemize}
\qed\end{proof}

The following example illustrates Proposition~\ref{prop:inclusion}. 

\begin{example}
Let $(S, \preceq)$ be a semilattice and let $(x, \{a\})$ and $(y, \{b_1,b_2\})$, with $x, y, a, b_1, b_2 \in S$, be two pseudo-elements as depicted in Figure~\ref{fig:pa_inclusion}. The pseudo-closure of $(x, \{a\})$ is depicted in dark gray, whereas the pseudo-closure of $(y, \{b_1,b_2\})$ is depicted in (light and dark) gray. We have $x \preceq y$, $b_1\sqcap x = b_1 \in$ $\antclos\{a\}$ and $b_2\sqcap x = b_2 \in$ $\antclos\{a\}$. Therefore $\pseudclos(x, \{a\}) \subseteq$ $\pseudclos(y, \{b_1,b_2\})$.
\end{example}

The next corollary is a direct consequence of the previous proposition.

\begin{corollary}
Let $(x, \alpha)$ and $(y, \beta)$ be two pseudo-elements in canonical form. Then $\pseudclos(x, \alpha)$ $=$ $\pseudclos(y, \beta)$ iff $x = y$ and $\alpha = \beta$.
\end{corollary}
\begin{proof}
We only prove $\pseudclos(x, \alpha) =$ $\pseudclos(y, \beta)  \Rightarrow x = y$ and $\alpha = \beta$, the other implication being trivial.

Since $\pseudclos(x, \alpha) =$ $\pseudclos(y, \beta)$, we have $\pseudclos(x, \alpha) \subseteq$ $\pseudclos(y, \beta)$ and $\pseudclos(y, \beta) \subseteq$ $\pseudclos(x, \alpha)$. By Proposition~\ref{prop:inclusion}, from $\pseudclos(x, \alpha) \subseteq$ $\pseudclos(y, \beta)$, we know that $x \preceq y$ and $\forall b \in \beta \cdot b \sqcap x \in$ $\antclos\alpha$, and from $\pseudclos(y, \beta) \subseteq$ $\pseudclos(x, \alpha)$, we know that $y \preceq x$ and $\forall a \in \alpha \cdot a \sqcap y \in$ $\antclos\beta$. As $x \preceq y$ and $y \preceq x$, we thus have $x = y$. 

Since $x = y$, by definition of canonical form of pseudo-elements, we have $\forall b \in \beta \cdot b \sqcap x = b \sqcap y = b \in$ $\antclos\alpha$ and $\forall a \in \alpha \cdot a \sqcap y = a \in$ $\antclos\beta$. It follows by Proposition~\ref{prop:opantichains} that $\antclos\alpha \subseteq$ $\antclos\beta$ and $\antclos\beta \subseteq$ $\antclos\alpha$, i.e. $\antclos\alpha =$ $\antclos\beta$. Since antichains are canonical representations of closed sets, we finally get $\alpha = \beta$, which terminates the proof.
\qed\end{proof}

\subsection{Pseudo-antichains}
We are now ready to introduce the new structure of pseudo-antichain. 
A \textit{pseudo-antichain} $A$ is a finite set of pseudo-elements, that is $A = \{(x_i, \alpha_i) \ | \ i \in I\}$ with $I$ finite.
The \textit{pseudo-closure} $\pseudclos A$ of $A$ is defined as the set $\pseudclos A = \bigcup_{i\in I}\pseudclos(x_i, \alpha_i)$. Let $(x_i, \alpha_i)$, $(x_j, \alpha_j) \in A$. We have the two following observations:
\begin{enumerate}
\item If $x_i = x_j$, then $(x_i, \alpha_i)$ and $(x_j, \alpha_j)$ can be replaced in $A$ by the pseudo-element $(x_i, \alpha_i\ \antinter\ \alpha_j)$.
\item If $\pseudclos(x_i, \alpha_i) \subseteq$ $\pseudclos(x_j, \alpha_j)$, then $(x_i, \alpha_i)$ can be removed from $A$.
\end{enumerate}
From these observations, we say that a pseudo-antichain $A = \{(x_i, \alpha_i) \ | \ i \in I\}$ is \textit{simplified} if $\forall i \cdot (x_i, \alpha_i)$ is in canonical form, and $\forall i \neq j \cdot x_i \neq x_j$ and $\pseudclos(x_i, \alpha_i) \not\subseteq$ $\pseudclos(x_j, \alpha_j)$. 
Notice that two distinct pseudo-antichains $A$ and $B$ can have the same pseudo-closure $\pseudclos A =$ $\pseudclos B$ even if they are simplified. We thus say that $A$ is a \textit{\PA-representation}\footnote{``\PA-representation" means pseudo-antichain based representation.} of $\pseudclos A$ (without saying that it is a canonical representation), and that $\pseudclos A$ is \textit{\PA-represented} by $A$. For efficiency purposes, our algorithms always work on simplified pseudo-antichains.

Any antichain $\alpha$ can be seen as the pseudo-antichain $A = \{(x, \emptyset) \mid x \in \alpha\}$. Furthermore, notice that \textit{any} set $X$ can be represented by the pseudo-antichain $A = \{(x, \alpha_x) \ | \ x \in X\}$, with  $\alpha_x = \lceil \{s \in S \ | \ s \preceq x \textnormal{ and } s \neq x\}\rceil$. Indeed $\pseudclos (x, \alpha_x) =  \{x\}$ for all $x$, and thus $X =$ $\pseudclos A$. 

The interest of pseudo-antichains is that given two antichains $\alpha$ and $\beta$, the difference $\antclos\alpha \diff\!\!\antclos\beta$ is \PA-represented by the pseudo-antichain $\{(x, \beta) \ | \ x \in \alpha\}$. 

\begin{lemma}\label{lem:diff} 
Let $\alpha, \beta \subseteq S$ be two antichains. Then
$\antclos\alpha \diff\!\!\antclos\beta =\ \pseudclos\{(x, \beta) \ | \ x \in \alpha\}$.
\end{lemma}

The next proposition indicates how to compute pseudo-closures of pseudo-elements w.r.t. the union, intersection and difference operations. This method can be extended for computing the union, intersection and difference of pseudo-closures of pseudo-antichains, by using the classical properties from set theory like
$X \diff (Y \cup Z) =  X \diff Y \cap X \diff Z$. From the algorithmic point of view, it is important to note that the computations only manipulate (pseudo-)antichains instead of their (pseudo-)closure.

\begin{proposition} 
\label{prop:BooleanOp}
Let $(x, \alpha), (y, \beta)$ be two pseudo-elements.
Then:
\begin{itemize}
\item $\pseudclos(x, \alpha)$ $\cup$ $\pseudclos(y, \beta) =$ $\pseudclos\{(x, \alpha), (y, \beta)\}$
\item $\pseudclos(x, \alpha)$ $\cap$ $\pseudclos(y, \beta) =$ $\pseudclos\{(x\sqcap y, \alpha$ $\antunion$ $\beta)\}$
\item $\pseudclos(x, \alpha)$ $\diff$ $\pseudclos(y, \beta) =$ $\pseudclos ( \{(x, \{y\}$ $\antunion$ $\alpha) \} \cup \{ (x\sqcap b, \alpha) \mid b \in \beta \} )$
\end{itemize}
\end{proposition}

Notice that there is an abuse of notation in the previous proposition. Indeed, the definition of a pseudo-element $(x,\alpha)$ requires that $x \not\in$ $\antclos\alpha$, whereas this condition could not be satisfied by couples like $(x\sqcap y, \alpha$ $\antunion$ $\beta)$, $(x, \{y\}$ $\antunion$ $\alpha)$ and $(x\sqcap b, \alpha) $ in the previous definition\footnote{This can be easily tested by Proposition~\ref{prop:opantichains}.}. When this happens, such a couple should not be added to the related pseudo-antichain. For instance, $\pseudclos(x, \alpha)$ $\cap$ $\pseudclos(y, \beta)$ is either equal to $\pseudclos\{(x\sqcap y, \alpha$ $\antunion$ $\beta)\}$ or to $\pseudclos\{\}$. 
Notice also that the pseudo-antichains computed in the previous proposition are not necessarily simplified. However, our algorithms implementing those operations always simplify the computed pseudo-antichains for the sake of efficiency.

\begin{proof}[of Proposition~\ref{prop:BooleanOp}]
We prove the three statements:
\begin{itemize}
\item $\pseudclos(x, \alpha)$ $\cup$ $\pseudclos(y, \beta) =$ $\pseudclos\{(x, \alpha), (y, \beta)\}$:\\
This result comes directly from the definition of pseudo-closure of pseudo-antichains.\\

\item $\pseudclos(x, \alpha)$ $\cap$ $\pseudclos(y, \beta) =$ $\pseudclos\{(x\sqcap y, \alpha$ $\antunion$ $\beta)\}$:
\begin{equation*}
\begin{split}
s \in\ \pseudclos(x, \alpha)\ \cap \pseudclos(y, \beta) &\Leftrightarrow s \preceq x, s \not\in\ \antclos\alpha \textnormal{ and } s \preceq y, s \not\in\ \antclos \beta\\
&\Leftrightarrow s \preceq x \sqcap y \textnormal{ and } s \not\in\ \antclos\alpha\ \cup\antclos\beta =\ \antclos(\alpha\ \antunion\ \beta) \textnormal{ (by Proposition~\ref{prop:opantichains})}\\
&\Leftrightarrow s \in\ \pseudclos\{(x\sqcap y, \alpha\ \antunion\ \beta)\}
\end{split}
\end{equation*}
\

\item $\pseudclos(x, \alpha)$ $\diff$ $\pseudclos(y, \beta) =$ $\pseudclos ( \{(x, \{y\}$ $\antunion$ $\alpha) \} \cup \{ (x\sqcap b, \alpha) \mid b \in \beta \} )$:\\
We prove the two inclusions:
\begin{enumerate}
\item \underline{$\subseteq:$} Let $s\in$ $\pseudclos(x, \alpha)\diff\!\!\pseudclos(y, \beta)$, i.e. $s\in$ $\pseudclos(x, \alpha)$ and $s\not\in$ $\pseudclos(y, \beta)$. Then, $s\preceq x$, $s \not\in$ $\antclos\alpha$ and ($s \not\preceq y$ or $s \in$ $\antclos\beta$). Thus, if $s\not\preceq y$, then $s \in$ $\pseudclos(x, \{y\}\ \antunion\ \alpha)$. Otherwise, $s \in$ $\antclos\beta$, i.e. $\exists b \in \beta$ such that $s \preceq b$. It follows that $s \in$ $\pseudclos(x\sqcap b, \alpha)$.

\item \underline{$\supseteq:$} Let $s \in$ $\pseudclos ( \{(x, \{y\}$ $\antunion$ $\alpha) \} \cup \{ (x\sqcap b, \alpha) \mid b \in \beta \} )$. 
Suppose first that $s \in$ $\pseudclos(x, \{y\}\ \antunion\ \alpha)$. Then $s\preceq x$, $s \not\preceq y$ and $s \not\in$ $\antclos\alpha$. We thus have $s\in$ $\pseudclos(x, \alpha)$ and $s\not\in$ $\pseudclos(y, \beta)$.
Suppose now that $\exists b \in \beta \cdot s \in$ $\pseudclos(x\sqcap b, \alpha)$. We have $s\preceq x$, $s \preceq b$ and $s \not\in$ $\antclos\alpha$. It follows that $s\in$ $\pseudclos(x, \alpha)$ and $s \in$ $\antclos\beta$, thus $s\not\in$ $\pseudclos(y, \beta)$.
\end{enumerate}
\end{itemize}
\qed\end{proof}

The following example illustrates the second and third statements of Proposition~\ref{prop:BooleanOp}.
\begin{example}
Let $(S, \preceq)$ be a lower semilattice and let $(x, \{a\})$ and $(y, \{b\})$, with $x, y, a, b \in S$, be two pseudo-elements as depicted in  Figure~\ref{fig:op_pa}. We have $\pseudclos(x, \{a\})$ $\cap$ $\pseudclos(y, \{b\}) =$ $\pseudclos(x\sqcap y, \{a,b\})$. We also have $\pseudclos(x, \{a\})$ $\setminus$ $\pseudclos(y, \{b\}) =$ $\pseudclos\{(x, \{y\}\antunion\{a\}), (x \sqcap b, \{a\})\} =$ $\pseudclos\{(x, \{y\}), (b, \{a\})\}$. Note that $(x, \{y\})$ and $(b, \{a\})$ are not in canonical form. The canonical representation of $\pseudclos(x, \{y\})$ (resp. $\pseudclos (b, \{a\})$) is given by $(x, \{x\sqcap y\})$ (resp. $(b, \{b\sqcap a\})$).
\end{example}
\begin{figure}[h]
\centering
\vspace*{-0.6cm}
  \begin{tikzpicture}[-,>=stealth',auto,node distance=2.5cm,bend angle=45,scale=0.85,font=\footnotesize]
    \fill[light-gray] (-2,-1) -- (-0.5, 0.5) --(1.5,-1.5) -- (1,-2) -- (0, -1) -- (-0.5,-1.5) -- (-1,-1) -- (-1.5,-1.5);
    \fill[light-gray] (4.5,-0.5) -- (6, 1) --(6.5,0.5) -- (5,-1);
    \fill[light-gray] (7,-1) -- (8, -2) --(7.5,-2.5) -- (6.5,-1.5);
    \node (0)  at (-1, 1.3) {$x$};
    \node (0)  at (-1, -0.7) {$a$};
    \node (0)  at (0.5, 1.8) {$y$};
    \node (0)  at (0, -0.7) {$b$};
    \node (0)  at (-0.03, 0.48) {$x\!\sqcap\! y$};
    \node (0)  at (-1, 1) {$\bullet$};
    \node (0)  at (-0.5, 0.5) {$\bullet$};
    \node (0)  at (-1, -1) {$\bullet$};
    \node (0)  at (0.5, 1.5) {$\bullet$};
    \node (0)  at (0, -1) {$\bullet$};

    \node (0)  at (6, 1.3) {$x$};
    \node (0)  at (6, -0.7) {$a$};
    \node (0)  at (7.5, 1.8) {$y$};
    \node (0)  at (7, -0.7) {$b$};
    \node (0)  at (6, 1) {$\bullet$};
    \node (0)  at (6, -1) {$\bullet$};
    \node (0)  at (7.5, 1.5) {$\bullet$};
    \node (0)  at (7, -1) {$\bullet$};
\path
	(-2.5,-0.5) edge (-1,1)
	(-1,1) edge (1.5,-1.5)
	(-1,-1) edge (-1.5,-1.5)
	(-1,-1) edge (0.5,-2.5)
	(0,-1) edge (-1,-2)
	(0,-1) edge (1,-2)
	(0.5,1.5) edge (2.5,-0.5)
	(0.5,1.5) edge (-2,-1)

	(4.5,-0.5) edge (6,1)
	(6,1) edge (8.5,-1.5)
	(6,-1) edge (5.5,-1.5)
	(6,-1) edge (7.5,-2.5)
	(7,-1) edge (6,-2)
	(7,-1) edge (8,-2)
	(7.5,1.5) edge (9.5,-0.5)
	(7.5,1.5) edge (5,-1)
	
	(-2.5,-0.5) edge[ultra thin,dashed] (0.5,2.5)
	(-2,-1) edge[ultra thin,dashed] (1,2)
	(-1.5,-1.5) edge[ultra thin,dashed] (1.5,1.5)
	(-1,-2) edge[ultra thin,dashed] (2,1)
	(-0.5,-2.5) edge[ultra thin,dashed] (2.5,0.5)
	(0,-3) edge[ultra thin,dashed] (3, 0)
	(0,-3) edge[ultra thin,dashed] (-3,0)
	(0.5,-2.5) edge[ultra thin,dashed] (-2.5,0.5)
	(1,-2) edge[ultra thin,dashed] (-2,1)
	(1.5,-1.5) edge[ultra thin,dashed] (-1.5,1.5)
	(2,-1) edge[ultra thin,dashed] (-1,2)
	(2.5,-0.5) edge[ultra thin,dashed] (-0.5,2.5)

	(4.5,-0.5) edge[ultra thin,dashed] (7.5,2.5)
	(5,-1) edge[ultra thin,dashed] (8,2)
	(5.5,-1.5) edge[ultra thin,dashed] (8.5,1.5)
	(6,-2) edge[ultra thin,dashed] (9,1)
	(6.5,-2.5) edge[ultra thin,dashed] (9.5,0.5)
	(7,-3) edge[ultra thin,dashed] (10,0)
	(7,-3) edge[ultra thin,dashed] (4,0)
	(7.5,-2.5) edge[ultra thin,dashed] (4.5,0.5)
	(8,-2) edge[ultra thin,dashed] (5,1)
	(8.5,-1.5) edge[ultra thin,dashed] (5.5,1.5)
	(9,-1) edge[ultra thin,dashed] (6,2)
	(9.5,-0.5) edge[ultra thin,dashed] (6.5,2.5);
\end{tikzpicture}
\vspace*{-0.2cm}
\caption{Intersection (left) and difference (right) of two pseudo-closures of pseudo-elements.}
\vspace*{-0.8cm}
\label{fig:op_pa}
\end{figure}
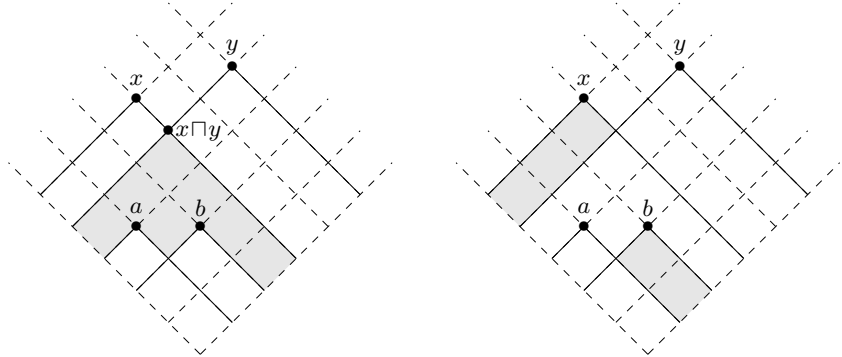
\section{Pseudo-antichain based algorithms}\label{sec:paalgo}
In this section, we propose a pseudo-antichain based version of the symblicit algorithm described in Section~\ref{sec:symblicit}  for solving the SSP and EMP problems for monotonic MDPs. In our approach, equivalence relations and induced partitions are symbolically represented so that each block is \PA-represented. The efficiency of this approach is thus directly linked to the number of blocks to represent, which explains why our algorithm always works with coarsest equivalence relations. It is also linked to the size of the pseudo-antichains representing the blocks of the partitions.

\subsection{Operator $\Pre_{\actionO,\actionI}$}
We begin by presenting an operator denoted $\Pre_{\actionO,\actionI}$ that is very useful for our algorithms.
Let $M_\preceq = (S, \ActionsO, \ActionsI, \edges, \distr)$ be a monotonic MDP equipped with a cost function $\reward : S \times \ActionsO \rightarrow \R$. 
Given $L \subseteq S$, $\actionO \in \ActionsO$ and $\actionI \in \ActionsI$, we denote by $\Pre_{\actionO,\actionI}(L)$ the set of states that reach $L$ by $\actionO,\actionI$ in $\monMDP$, that is 
$$\Pre_{\actionO,\actionI}(L) = \{s\in S \mid \edges(s,\actionO)(\actionI) \in L \}.$$ 
The elements of $\Pre_{\actionO,\actionI}(L)$ are called \textit{predecessors} of $L$ for $\actionO,\actionI$ in $\monMDP$. The following lemma is a direct consequence of the compatibility of $\preceq$.

\begin{lemma}\label{lem:closed}
For all closed set $L \subseteq S$, and all actions $\actionO \in \ActionsO, \actionI \in \ActionsI$, $\Pre_{\actionO,\actionI}(L)$ is closed.
\end{lemma}

The next lemma indicates the behavior of the $\Pre_{\actionO,\actionI}$ operator under boolean operations. The second and last properties follow from the fact that $\monMDP$ is deterministic. 

\begin{lemma}\label{lem:opboolpre}
Let $L_1, L_2 \subseteq S$, $\actionO \in \ActionsO$ and $\actionI \in \ActionsI$. Then,
\begin{itemize}
	\item $\Pre_{\actionO,\actionI}(L_1 \cup L_2) = \Pre_{\actionO,\actionI}(L_1) \cup \Pre_{\actionO,\actionI}(L_2)$ 
	\item $\Pre_{\actionO,\actionI}(L_1 \cap L_2) = \Pre_{\actionO,\actionI}(L_1) \cap \Pre_{\actionO,\actionI}(L_2)$ 
	\item $\Pre_{\actionO,\actionI}(L_1 \diff L_2) = \Pre_{\actionO,\actionI}(L_1) \diff \Pre_{\actionO,\actionI}(L_2)$ 
\end{itemize}
\end{lemma}

\begin{proof}
The first property is immediate. We only prove the second property, since the arguments are similar for the last one.
Let $s \in \Pre_{\actionO,\actionI}(L_1 \cap L_2)$, i.e. $\exists s' \in L_1\cap L_2$ such that $\edges(s, \actionO)(\actionI) = s'$. We thus have $s \in \Pre_{\actionO,\actionI}(L_1)$ and $s \in \Pre_{\actionO,\actionI}(L_2)$. 
Conversely let $s \in \Pre_{\actionO,\actionI}(L_1) \cap \Pre_{\actionO,\actionI}(L_2)$, i.e. $\exists s_1 \in L_1$ such that $\edges(s, \actionO)(\actionI) = s_1$, and $\exists s_2 \in L_2$ such that $\edges(s, \actionO)(\actionI) = s_2$. As $\monMDP$ is deterministic, we have $s_1 = s_2$ and thus $s \in \Pre_{\actionO,\actionI}(L_1 \cap L_2)$.
\qed\end{proof}

The next proposition indicates how to compute pseudo-antichains w.r.t. the $\Pre_{\actionO,\actionI}$ operator.
\begin{proposition}
\label{prop:Presigma}
Let $(x, \alpha)$ be a pseudo-element with $x \in S$ and $\alpha \subseteq S$. Let $A = \{(x_i, \alpha_i) \ | \ i \in I\}$ be a pseudo-antichain with $x_i \in S$ and $\alpha_i \subseteq S$ for all $i \in I$. Then, for all $\actionO \in \ActionsO$ and $\actionI \in \ActionsI$,
\begin{itemize}
	\item $\Pre_{\actionO,\actionI}(\pseudclos(x, \alpha)) = \bigcup_{x' \in \lceil \Pre_{\actionO,\actionI}(\antclos\ \{x\})\rceil}\pseudclos(x', \lceil\Pre_{\actionO,\actionI}(\antclos\alpha)\rceil)$
	\item $\Pre_{\actionO,\actionI}(\pseudclos A) = \bigcup_{i \in I}\Pre_{\actionO,\actionI}(\pseudclos (x_i, \alpha_i))$
\end{itemize}
\end{proposition}

\begin{proof}
For the first statement, we have $\Pre_{\actionO,\actionI}(\pseudclos(x, \alpha)) = \Pre_{\actionO,\actionI}(\antclos\{x\} \diff\!\antclos\alpha) = \Pre_{\actionO,\actionI}(\antclos\{x\}) \diff \Pre_{\actionO,\actionI}(\antclos\alpha)$ by definition of the pseudo-closure and by Lemma~\ref{lem:opboolpre}. The sets $\Pre_{\actionO,\actionI}(\antclos \{x\})$ and $\Pre_{\actionO,\actionI}(\antclos \alpha)$ are closed by Lemma~\ref{lem:closed} and thus respectively represented by the antichains $\lceil \Pre_{\actionO,\actionI}(\antclos\{x\})\rceil$ and $\lceil\Pre_{\actionO,\actionI}(\antclos\alpha)\rceil$. By Lemma~\ref{lem:diff} we get the first statement.

The second statement is a direct consequence of Lemma~\ref{lem:opboolpre}.
\qed\end{proof}

From Proposition~\ref{prop:Presigma}, we can efficiently compute pseudo-antichains w.r.t. the $\Pre_{\actionO,\actionI}$ operator if we have an efficient algorithm to compute antichains w.r.t. $\Pre_{\actionO,\actionI}$ (see the first statement). We make the following assumption that we can compute the predecessors of a closed set by only considering the antichain of its maximal elements. Together with Proposition~\ref{prop:Presigma}, it implies that the computation of $\Pre_{\actionO,\actionI}(\pseudclos A)$, for all pseudo-antichain $A$, does not need to treat the whole pseudo-closure $\pseudclos A$.

\begin{assumption} \label{atwo}
There exists an algorithm taking any state $x \in S$ in input and returning $\lceil \Pre_{\actionO,\actionI}(\antclos \{x\})\rceil$ as output.
\end{assumption}



\begin{remark}\label{rem:ass1}
Assumption~\ref{atwo} is a realistic and natural assumption when considering partially ordered state spaces. For instance, it holds for the two classes of MDPs considered in Section~\ref{sec:experiments} for which the given algorithm is straightforward. Assumptions in the same flavor are made in~\cite{DBLP:journals/tcs/FinkelS01} (see Definition 3.2).
\end{remark}

\subsection{Symbolic representations}  \label{subsec:symbrep}

Before giving a pseudo-antichain based algorithm for the symblicit approach of Section~\ref{sec:symblicit} (see Algorithm~\ref{algo:symblicit}), we detail in this section the kind of symbolic representations based on pseudo-antichains that we are going to use. Recall from Section~\ref{sec:pa} that PA-representations are not unique. For efficiency reasons, it will be necessary to work with PA-representations as \textit{compact} as possible, as suggested in the sequel.

\paragraph{Representation of the stochastic models.}
We begin with symbolic representations for the monotonic MDP $\monMDP = (S, \ActionsO, \ActionsI, \edges, \distr)$ and for the MC $\monMC = (S, \ActionsI, \edges_{\lambda}, \distr_{\lambda})$ induced by a strategy $\lambda$. 
For algorithmic purposes, in addition to Assumption~\ref{atwo}, we make the following assumption\footnote{Remark~\ref{rem:ass1} also holds for Assumption~\ref{aone}.} on $\monMDP$. 
\begin{assumption} \label{aone}
There exists an algorithm taking in input any state $s \in S$ and actions $\actionO \in \enabledactions, \actionI \in \ActionsI$, and returning as output $\edges(s, \actionO)(\actionI)$  and $\distr(s, \actionO)(\actionI)$.
\end{assumption}

By definition of $\monMDP$, the set $S$ of states is closed for $\preceq$ and can thus be canonically represented by the antichain $\lceil S \rceil$, and thus represented by the pseudo-antichain $\{(x,\emptyset) \mid x \in \lceil S \rceil\}$. In this way, it follows by Assumption~\ref{aone} that we have a \PA-representation of $\monMDP$, in the sense that $S$ is \PA-represented and we can compute $\edges(s, \actionO)(\actionI)$  and $\distr(s, \actionO)(\actionI)$ on demand.

Let $\lambda : S \rightarrow \ActionsO$ be a strategy on $\monMDP$ and $\monMC$
be the induced MC with cost function $\reward_\lambda$. 
We denote by $\equivstrat$ the equivalence relation on $S$ such that $s \equivstrat s'$ iff $\lambda(s) = \lambda(s')$. We denote by $\partstrat$ the induced partition of $S$. Given a block $B \in \partstrat$, we denote by $\lambda(B)$ the unique action $\lambda(s)$, for all $s \in B$. As any set can be represented by a pseudo-antichain, each block of $\partstrat$ is \PA-represented. Therefore by Assumption~\ref{aone}, we have a \PA-representation of $\monMC$. 

\paragraph{Representation of a subset of goal states.}
Recall that a subset $G \subseteq S$ of goal states is required for the SSP problem. Our algorithm will manipulate $G$ when computing the set of proper states. A natural assumption is to require that $G$ is \textit{closed} (like $S$), as it is the case for the two classes of monotonic MDPs studied in Section~\ref{sec:experiments}. Under this assumption, we have a compact representation of $G$ as the one proposed above for $S$.   
Otherwise, we take for $G$ any \PA-representation.

\paragraph{Representation for $\distr$ and $\reward$.}
For the needs of our algorithm, we introduce symbolic representations for $\distr_\lambda$ and $\reward_\lambda$.
Similarly to $\equivstrat$,  let $\equivdistr$ be the equivalence relation on $S$ such that $s \equivdistr s'$ iff $\distr_\lambda(s) = \distr_\lambda(s')$. We denote by $\partdistr$ the induced partition of $S$. Given a block $B \in \partdistr$, we denote by $\distr_\lambda(B)$  the unique probability distribution $\distr_\lambda(s)$, for all $s \in B$. We use similar notations for the equivalence relation $\equivreward$ on $S$ such that $s \equivreward s'$ iff $\reward_\lambda(s) = \reward_\lambda(s')$.
As any set can be represented by a pseudo-antichain, each block of $\partdistr$ and $\partreward$ is \PA-represented.

We will also need to use the next two equivalence relations. For each $\actionO \in \ActionsO$, we introduce the equivalence relation $\sim_{\distr,\sigma}$ on $S$ such that $s \sim_{\distr,\sigma} s'$ iff $\distr(s,\sigma) = \distr(s',\sigma)$.  Similarly, we introduce relation $\equivsigmarew$ such that $s \equivsigmarew s'$ iff $\reward(s, \actionO) = \reward(s', \actionO)$.  Recall that $\distr$ and $\reward$ are partial functions, there may thus exist one block in their corresponding relation gathering all states $s$ such that $\actionO \not\in \ActionsO_s$. Each block of the induced partitions $S_{\sim_{\distr,\sigma}}$ and $S_{\sim_{\reward,\sigma}}$  is PA-represented.

For the two classes of MDPs studied in Section~\ref{sec:experiments}, both functions $\distr$ and $\reward$ are independent of $S$. It follows that the previous equivalence relations have only one or two blocks, leading to compact symbolic representations of these relations. 

\bigskip
Now that the operator $\Pre_{\actionO,\actionI}$ and the used symbolic representations have been introduced, we come back to the different steps of the symblicit approach of Section~\ref{sec:symblicit} (see Algorithm~\ref{algo:symblicit}) and show how to derive a pseudo-antichain based algorithm. We will use Propositions~\ref{prop:opantichains}, \ref{prop:BooleanOp} and~\ref{prop:Presigma}, and Assumptions~\ref{atwo} and~\ref{aone}, for which we know that boolean and $\Pre_{\actionO,\actionI}$ operations can be performed efficiently on pseudo-closures of pseudo-antichains, by limiting the computations to the related pseudo-antichains. Whenever possible, we will work with partitions with few blocks whose PA-representation is compact. This aim will be reached for the two classes of monotonic MDPs studied in Section~\ref{sec:experiments}. 

\subsection{Initial strategy} 
Algorithm~\ref{algo:symblicit} needs an initial strategy $\lambda_0$ (line 1). This strategy can be selected arbitrarily among the set of strategies for the EMP, while it has to be a proper strategy for the SSP. We detail how to choose the initial strategy in these two quantitative settings.

\paragraph{Expected mean-payoff.} For the EMP, we propose an arbitrary initial strategy $\lambda_0$ with a compact \PA-representation for the induced MC $\monMCzero$. We know that $S$ is \PA-represented by $\{(x,\emptyset) \mid x \in \lceil S \rceil\}$, and that for all $s,s' \in S$ such that $s \preceq s'$, we have $\enabledactionsprime \subseteq \enabledactions$ (Proposition~\ref{prop:As}). This means that for $x \in \lceil S \rceil$ and  $\actionO \in \ActionsO_{x}$ we could choose $\lambda_0(s) = \actionO$ for all $s \in \antclos \{x\}$. However we must be careful with states that belong to $\antclos \{x\} ~\cap \antclos \{x'\}$ with $x, x' \in \lceil S \rceil$, $x \neq x'$. Therefore, let us impose an arbitrary ordering on $\lceil S \rceil\ = \{x_1, \ldots, x_n\}$, i.e. $x_1 < x_2 < \ldots < x_n$. We then define $\lambda_0$ arbitrarily on $\lceil S \rceil$ such that $\lambda_0(x_i) = \actionO_i$ for some $\actionO_i \in \enabledactionsxi$, and we extend it to all $s \in S$ by $\lambda_0(s) = \lambda_0(x)$ with $x = \min_i\{x_i \mid s \preceq x_i\}$. This makes sense in view of the previous remarks. Notice that given $\actionO \in \ActionsO$, the block $B$ of the partition $S_{\equivstratzero}$ such that $\lambda_0(B) = \actionO$ is \PA-represented by $\bigcup_i\{ (x_i,\alpha_i) \mid \alpha_i = \{x_1, \ldots, x_{i-1}\}, \lambda_0(x_i) = \actionO \}$.

\paragraph{Proper states.}
Before explaining how to compute an initial proper strategy $\lambda_0$ for the SSP, we need to propose a pseudo-antichain based version of the algorithm of~\cite{DBLP:conf/concur/Alfaro99} for computing the set $S^P$ of proper states. Recall from Section~\ref{subsec:SSP} that this algorithm is required for solving the SSP problem. 

Let $M_\preceq$ be a monotonic MDP and $G$ be a set of goal states. Recall that $S^P$ is computed as $S^P = \nu Y \cdot \mu X \cdot (\apre(Y,X) \lor \mathsf{G})$, such that for all state~$s$, 
$$s \models \apre(Y, X) \Leftrightarrow \exists \actionO \in \ActionsO_s, (\suc(s, \actionO) \subseteq Y \land \suc(s,\actionO) \cap X \neq \emptyset).$$
Our purpose is to define the set of states satisfying $\apre(Y, X)$ thanks to the operator $\Pre_{\actionO,\actionI}$. The difficulty is to limit the computations to strictly positive probabilities as required by the operator $\suc$. In this aim, given the equivalence relation $\sim_{\distr,\sigma}$ defined in Section~\ref{subsec:symbrep}, for each $\actionO \in \ActionsO$ and $\actionI \in \ActionsI$, we define ${\cal D}^{>0}_{\actionO,\actionI}$ being the set of blocks $\{D \in S_{\sim_{\distr,\sigma}} \mid  \distr(s,\actionO)(\actionI) > 0 \mbox{ with } s \in D\}$. For each $D \in {\cal D}^{>0}_{\actionO,\actionI}$, notice that $\actionO \in \ActionsO_s$ for all $s \in D$ (since $\distr(s,\actionO)$ is defined).
Given two sets $X, Y \subseteq S$, the set of states satisfying $\apre(Y, X)$ is equal to:
$$R(Y,X) = \bigcup_{\actionO \in \ActionsO} \bigcup_{D \in S_{\sim_{\distr,\sigma}}} \Big( \bigcap_{\begin{smallmatrix}\actionI \in \ActionsI\\D \in {\cal D}^{>0}_{\actionO,\actionI}\end{smallmatrix}}(\Pre_{\actionO,\actionI}(Y) \cap D) \cap \bigcup_{\begin{smallmatrix}\actionI \in \ActionsI\\D \in {\cal D}^{>0}_{\actionO,\actionI}\end{smallmatrix}} (\Pre_{\actionO,\actionI}(X) \cap D) \Big)$$

\begin{lemma}
For all $X, Y \subseteq S$ and $s \in S$, $s \models \apre(Y, X) \Leftrightarrow s \in R(Y,X)$.
\end{lemma}

\begin{proof}
Let $s \models \apre(Y, X)$. Then there exists $\actionO \in \ActionsO_s$ such that $\suc(s, \actionO) \subseteq Y$ and  $\suc(s,\actionO) \cap X \neq \emptyset$. Let $D \in S_{\sim_{\distr,\sigma}}$ be such that $s \in D$. Let us prove that $s \in  \bigcap_{\actionI \in \ActionsI, D \in {\cal D}^{>0}_{\actionO,\actionI}}\Pre_{\actionO,\actionI}(Y)$ and $s \in  \bigcup_{\actionI \in \ActionsI, D \in {\cal D}^{>0}_{\actionO,\actionI}} \Pre_{\actionO,\actionI}(X)$. It will follows that $s \in R(Y,X)$. As $\suc(s,\actionO) \cap X \neq \emptyset$, there exists $x \in X$ such that $\probmat(s,\sigma,x) > 0$, that is, $\edges(s,\actionO)(\actionI) = x$ and $\distr(s,\actionO)(\actionI) > 0$ for some $\actionI \in \ActionsI$. Thus $s \in \Pre_{\actionO,\actionI}(X)$ and $D \in {\cal D}^{>0}_{\actionO,\actionI}$. As $\suc(s, \actionO) \subseteq Y$, then for all $s'$ such that $\probmat(s,\sigma,s') > 0$, we have $s' \in Y$, or equivalently, for all $\actionI \in \ActionsI$ such that $\distr(s,\actionO)(\actionI) > 0$, we have $\edges(s,\actionO)(\actionI) = s' \in Y$. Therefore, $s \in \Pre_{\actionO,\actionI}(Y)$ for all $\actionI$ such that $D \in {\cal D}^{>0}_{\actionO,\actionI}$.

Suppose now that $s \in R(Y,X)$. Then we have that there exists $\actionO \in \ActionsO$ and  $D \in S_{\sim_{\distr,\sigma}}$ such that $s \in  \bigcap_{\actionI \in \ActionsI, D \in {\cal D}^{>0}_{\actionO,\actionI}}\Pre_{\actionO,\actionI}(Y)$ and $s \in  \bigcup_{\actionI \in \ActionsI, D \in {\cal D}^{>0}_{\actionO,\actionI}} \Pre_{\actionO,\actionI}(X)$. With the same arguments as above, we deduce that $\suc(s, \actionO) \subseteq Y$ and  $\suc(s,\actionO) \cap X \neq \emptyset$. Notice that we have $\actionO \in \ActionsO_s$. It follows that $s \models \apre(Y, X)$.
\qed\end{proof}

In the case of the MDPs treated in Section~\ref{sec:experiments}, $G$ is closed and ${\cal D}^{>0}_{\actionO,\actionI}$ is composed of at most one block $D$ that is closed. It follows that all the intermediate sets manipulated by the algorithm computing $S^P$ are closed. We thus have an efficient algorithm since it can be based on antichains only. 

\paragraph{Stochastic shortest path.} For the SSP, the initial strategy $\lambda_0$ must be proper. In the previous paragraph, we have presented an 
algorithm for computing the set $S^P = \nu Y \cdot \mu X \cdot (\apre(Y,X) \lor \mathsf{G})$ of proper states. One can directly extract a \PA-representation of a proper strategy from the execution of this algorithm, as follows. During the greatest fix point computation, a sequence $Y_0, Y_1, \ldots, Y_k$ is computed such that $Y_0 = S$ and $Y_{k-1} = Y_k = S^p$. During the computation of $Y_k$, a least fix point computation is performed, leading to a sequence $X_0, X_1, \ldots,  X_l$ such that $X_0 = G$ and $X_i = X_{i-1} \cup \apre(S^P, X_{i-1})$ for all $i$, $1 \leq i \leq l$, and $X_l = Y_k$. We define a strategy $\lambda_0$ incrementally with each set $X_i$. Initially, $\lambda_0(s)$ is any $\actionO \in \Sigma_s$ for each $s \in X_0$ (since $G$ is reached). When $X_i$ has been computed, then for each $s \in X_i \diff X_{i-1}$, $\lambda_0(s)$ is any $\actionO \in \Sigma_s$ such that $\suc(s, \actionO) \subseteq S^P$ and $\suc(s, \actionO) \cap X_{i-1} \neq \emptyset$. Doing so, each proper state is eventually associated with an action by $\lambda_0$. Note that this strategy is proper. Indeed, to simplify the argument, suppose $G$ is absorbing\footnote{for all $s \in G$ and $\actionO \in \ActionsO_s$, $\sum_{s'\in G}\probmat(s, \actionO, s') = 1$.}. Then by construction of $\lambda_0$, the bottom strongly connected components of the MC induced by $\lambda_0$ are all in $G$, and a classical result on MCs (see \cite{DBLP:books/daglib/0020348}) states that any infinite path will almost surely lead to one of those components. Finally, as all sets $X_i, Y_j$ manipulated by the algorithm are \PA-represented, we obtain a partition $S_{\equivstratzero}$ such that each block $B \in S_{\equivstratzero}$ is \PA-represented.

\subsection{Bisimulation lumping} 
We now consider the step of Algorithm~\ref{algo:symblicit} where Algorithm~\textsc{Lump} is called to compute the largest bisimulation $\equivlump$ of an MC $\monMC$
induced by a strategy $\lambda$ on $\monMDP$ (line 4 with $\lambda = \lambda_n$). We here detail a pseudo-antichain based version of Algorithm~\textsc{Lump} (see Algorithm~\ref{algo:lump}) when $\monMC$ is \PA-represented. Recall that if $\monMDP = (S, \ActionsO, \ActionsI, \edges, \distr)$, then $\monMC = (S, \ActionsI, \edges_{\lambda}, \distr_{\lambda})$. Equivalently, using the usual definition of an MC, $\monMC = (S, \probmat_\lambda)$ with $\probmat_\lambda$ derived from $\edges_{\lambda}$ and $\distr_{\lambda}$ (see Section~\ref{sec:prelim}). Remember also the two equivalence relations $\equivdistr$ and $\equivreward$ defined in Section~\ref{subsec:symbrep}. 

The initial partition $P$ computed by Algorithm~\ref{algo:lump} (line 1) is such that for all $s, s' \in S$, $s$ and $s'$ belong to the same block of $P$ iff $\reward_\lambda(s) = \reward_\lambda(s')$. The initial partition $P$ is thus $\partreward$.

Algorithm~\ref{algo:lump} needs to split blocks of partition $P$ (line 7). This can be performed thanks to Algorithm~\textsc{Split} that we are going to describe (see Algorithm~\ref{algo:split} below). 
Given two blocks $B, C \subseteq S$, this algorithm splits $B$ into a partition $P$ composed of sub-blocks $B_1, \ldots, B_k$ according to the probability of reaching $C$ in $\monMC$, i.e. for all $s,s' \in B$, we have $s,s' \in B_l$ for some $l$ iff $\probmat_\lambda(s,C) = \probmat_\lambda(s',C)$. 

\begin{algorithm}[t]
\caption{\textsc{Split}$(B, C, \lambda)$}
\begin{algorithmic}[1] \label{algo:split}
\STATE $P[0] := B$
\FOR{$i$ in $[1, m]$}
	\STATE $P_\textnormal{new}  :=$ \textsc{InitTable}$(P,\actionI_i)$
	\FORALL{$(p, \val)$ in $P$}
		\STATE $P_\textnormal{new}[p] := P_\textnormal{new}[p] \cup (\val \diff \PreMC(C, \actionI_i))$
		\FORALL{$D \in \partdistr$}
			\STATE $P_\textnormal{new}[p + \distr_\lambda(D)(\actionI_i)] := P_\textnormal{new}[p + \distr_\lambda(D)(\actionI_i)] \cup (\val \cap D \cap \PreMC(C, \actionI_i))$
		\ENDFOR
	\ENDFOR
	\STATE $P:=$ \textsc{RemoveEmptyBlocks}$(P_\textnormal{new})$ 
	\ENDFOR
\STATE return $P$
\end{algorithmic}
\end{algorithm}

Suppose that $\ActionsI = \{\actionI_1, \ldots, \actionI_m\}$. This algorithm computes intermediate partitions $P$ of $B$ such that at step~$i$, $B$ is split according to the probability of reaching $C$ in $\monMC$ when $\ActionsI$ is restricted to $\{\actionI_1, \ldots, \actionI_i\}$. To perform this task, it needs a new operator $\PreMC$ based on $\monMDP$ and $\lambda$. Given $L \subseteq S$ and $\actionI \in \ActionsI$, we define 
$$\PreMC(L, \actionI) = \{s \in S \mid \edges_{\lambda}(s)(\actionI) \in L\}$$ 
as the set of states from which $L$ is reached by $\actionI$ in $\monMDP$ under the selection made by $\lambda$. Notice that when $\ActionsI$ is restricted to $\{\actionI_1, \ldots, \actionI_i\}$ with $i < m$, it may happen that $\distr_\lambda(s)$ is no longer a probability distribution for some $s\in S$ (when $\sum_{\actionI \in \{\actionI_1, \ldots, \actionI_i\}}\distr_\lambda(s)(\actionI) < 1$).

Initially, $\ActionsI$ is restricted to $\emptyset$, and the partition $P$ is composed of one block $B$ (see line 1).
At step $i$ with $i \geq 1$, each block $B_l$ of the partition computed at step $i-1$ is split into several sub-blocks according to its intersection with $\PreMC(C,\actionI_i)$ and each $D \in \partdistr$. We take into account intersections with $D \in \partdistr$ in a way to know which stochastic function $\distr_\lambda(D)$ is associated with the states we are considering. Suppose that at step $i-1$ the probability for any state of block $B_l$ of reaching $C$ is $p$. Then at step $i$, it is equal to $p + \distr_\lambda(D)(\actionI_i)$ if this state belongs to $D \cap \PreMC(C,\actionI_i)$, with $D \in \partdistr$, and to $p$ if it does not belong to $\PreMC(C,\actionI_i)$ (lines 5-7). See Figure~\ref{fig:intuitionlump} for intuition. Notice that some newly created sub-blocks could have the same probability, they are therefore merged.

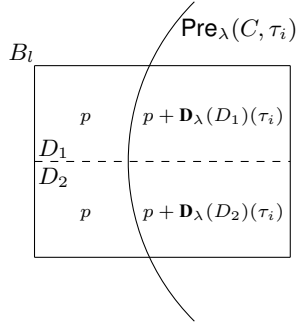
\begin{figure}
\vspace*{-0.9cm}
\centering
  \begin{tikzpicture}[-,>=stealth',auto,node distance=2.5cm,bend angle=45,scale=0.85,font=\footnotesize]
    \node (0)  at (0.8, 0.2) {$B_l$};
    \node (0)  at (1.3, -1.3) {$D_1$};
    \node (0)  at (1.3, -1.75) {$D_2$};
    \node (0)  at (4.2, 0.55) {$\PreMC(C,\actionI_i)$};
    \node (0)  at (1.8, -0.835) {{\scriptsize $p$}};
    \node (0)  at (1.8, -2.335) {{\scriptsize $p$}};
    \node (0)  at (3.8, -0.8) {{\scriptsize $p + \distr_\lambda(D_1)(\actionI_i)$}};
    \node (0)  at (3.8, -2.3) {{\scriptsize $p + \distr_\lambda(D_2)(\actionI_i)$}};
     \path
	(1,0) edge (5,0)
	(1,0) edge (1,-3)
	(1,-3) edge (5,-3)
	(5,-3) edge (5, 0)
	(1, -1.5) edge[dashed] (5, -1.5);
        \draw (3.5, 1) to[out=225,in=135] node {} (3.5,-4);
\end{tikzpicture}
\vspace*{-0.4cm}
\caption{Step $i$ of Algorithm~\ref{algo:split} on a block $B_l$.}
\label{fig:intuitionlump}
\vspace*{-0.3cm}
\end{figure}

The intermediate partitions $P$ (or $P_\textnormal{new}$) manipulated by the algorithm are represented by hash tables:  each entry $(p, \val)$ is stored as $P[p] = \val$ such that $\val$ is the set of states that reach $C$ with probability $p$.
The use of hash tables permits to efficiently gather sub-blocks of states having the same probability of reaching $C$, and thus to keep 
minimal the number of blocks in the partition. Algorithm~\textsc{InitTable} is used to initialize a new partition $P_\textnormal{new}$ from a previous partition $P$ and symbol $\actionI_i$: the new hash table is initialized with $P_\textnormal{new}[p] := \emptyset$ and $P_\textnormal{new}[p + \distr_\lambda(D)(\actionI_i)] := \emptyset$, for all $D \in \partdistr$ and all $(p, \val)$ in $P$. Algorithm \textsc{RemoveEmptyBlocks}$(P)$ removes from the hash table $P$ each pair $(p, block)$ such that $block = \emptyset$.

\begin{theorem}
\label{thm:correct}
Let $\lambda$ be a strategy on $\monMDP$ and $\monMC = (S, \probmat_\lambda)$ be the induced MC. Let $B, C \subseteq S$ be two blocks. Then the output of $\textsc{Split}(B,C,\lambda)$ is a partition $\{B_1,\ldots, B_k\}$ of $B$ such that for all $s,s' \in B$, $s,s' \in B_l$ for some $l$ iff $\probmat_\lambda(s,C) = \probmat_\lambda(s',C)$.
\end{theorem}

\begin{proof}
The correctness of Algorithm \textsc{Split} is based on the following invariant. At step~$i$, $0 \leq i \leq m$, with $\ActionsI$ restricted to $\{\actionI_1,\dots,\actionI_i\}$, we have:
\begin{itemize}
\item $P$ is a partition of $B$
\item $\forall s \in B$, if $s \in P[p]$, then $\probmat_\lambda(s, C) = p$
\end{itemize}

Let us first prove that $P$ is a partition. Note that the use of algorithm \textsc{RemoveEmptyBlocks} ensures that $P$ never contains an empty block. Recall that $\partdistr$ is a partition of $S$.

Initially, when $i = 0$, $P$ is composed of the unique block $B$. Let $i \geq 1$ and suppose that $P = \{B_1,\ldots, B_k\}$ is a partition of $B$ at step $i-1$, and let us prove that $P_\textnormal{new}$ is a partition of $B$ (see line 7 of the algorithm). Each $B_l \in P$ is partitioned as $\{B_l\diff\PreMC(C,\actionI_i)\} \cup \{B_l\cap D \cap\PreMC(C, \actionI_i)\mid D \in \partdistr\}$. This leads to the finer partition $P' = \{B_l \diff\PreMC(C,\actionI_i) \mid B_l \in P\} \cup \{B_l \cap D \cap \PreMC(C,\actionI_i) \mid B_l \in P, D \in \partdistr\}$ of $P$. Some blocks of $P'$ are gathered by the algorithm to get $P_\textnormal{new}$ which is thus a partition.

Let us now prove that at each step $i$, with $\ActionsI$ restricted to $\{\actionI_1,\dots,\actionI_i\}$, we have: $\forall s \in B$, if $s \in P[p]$, then $\probmat_\lambda(s, C) = p$.

Initially, when $i = 0$, $\ActionsI$ is restricted to $\emptyset$, and thus $\probmat_\lambda(s, C) = 0$ for all $s \in B$.
Let $i \geq 1$ and suppose we have that $\forall s \in B$, if $s \in P[p]$ for some $p$, then $\probmat_\lambda(s, C) = p$, when $\ActionsI$ is restricted to $\{\actionI_1,\dots,\actionI_{i-1}\}$. Let us prove that if $s \in P_\textnormal{new}[p]$ for some $p$, then $\probmat_\lambda(s, C) = p$, when $\ActionsI$ is restricted to $\{\actionI_1,\dots,\actionI_{i}\}$.
Let $s \in B$ be such that $s \in P_\textnormal{new}[p]$, we identify two cases:  either $(i)$ $s \in P[p]$ and $s \not\in \PreMC(C, \actionI_i)$, or $(ii)$ $s \in P[p-\distr_\lambda(D)(\actionI_i)]$ and $s \in D \cap \PreMC(C, \actionI_i)$, for some $D \in \partdistr$. In case $(i)$, by induction hypothesis, we know that $\probmat_\lambda(s, C) = p$, when $\ActionsI = \{\actionI_1,\dots,\actionI_{i-1}\}$, and since $s \not\in \PreMC(C, \actionI_i)$, adding $\actionI_i$ to $\ActionsI$ does not change the probability of reaching $C$ from $s$, i.e. $\probmat_\lambda(s, C) = p$ when $\ActionsI = \{\actionI_1, \dots, \actionI_i\}$. In case $(ii)$, by induction hypothesis, we know that $\probmat_\lambda(s, C) = p-\distr_\lambda(D)(\actionI_i)$, when $\ActionsI = \{\actionI_1, \dots, \actionI_{i-1}\}$. Moreover, since $s \in D \cap \PreMC(C, \actionI_i)$, we have $\probmat_\lambda(s, C) = \distr_\lambda(D)(\actionI_i)$, when $\ActionsI = \{\actionI_{i}\}$. It follows that $\probmat_\lambda(s, C) = p-\distr_\lambda(D)(\actionI_i)+\distr_\lambda(D)(\actionI_i) = p$, when $\ActionsI = \{\actionI_1, \dots, \actionI_{i}\}$.

Finally, after step $i=m$, we get the statement of Theorem~\ref{thm:correct} since $P$ is the final partition of $B$ and for all $s \in B$, $s \in P[p]$ iff $\probmat_\lambda(s, C) = p$.\footnote{the ``iff" holds since probabilities $p$ are pairwise distinct.}
\qed\end{proof}

Notice that we have a pseudo-antichain version of Algorithm~\textsc{Lump} as soon as the given blocks $B$ and $C$ are PA-represented. Indeed, this algorithm uses boolean operations and $\PreMC$ operator. This operator can be computed as follows:
\begin{equation*}
\PreMC(C, \actionI) = \bigcup \big\{ \Pre_{\actionO,\actionI}(C)\ \cap B \mid \actionO \in \ActionsO, B \in \partstrat, \lambda(B) = \actionO \big\}.
\end{equation*}
Intuitively, let us fix $\actionO \in \ActionsO$ and $B \in \partstrat$ such that $\lambda(B) = \actionO$. Then $\Pre_{\actionO,\actionI}(C)\cap B$ is the set of states of $S \cap B$ that reach $C$ in the MDP $\monMDP$ with $\actionO$ followed by $\actionI$. Finally the union gives the set of states that reach $C$ with $\actionI$ under the selection made by $\lambda$. All these operations can be performed thanks to Propositions~\ref{prop:BooleanOp}, \ref{prop:Presigma}, and Assumptions~\ref{atwo}, \ref{aone}.

\subsection{Solving linear systems} 

Assume that we have computed the largest bisimulation $\equivlump$ for the MC $\monMC = (S, \probmat_\lambda)$. Let us now detail lines 5-7 of Algorithm~\ref{algo:symblicit} concerning the system of linear equations that has to be solved. We build the Markov chain that is the bisimulation quotient $(\partlump, \probmat_{\lambda,\equivlump})$. We then explicitly solve the linear system of Algorithm~\ref{algo:SSPSI} for the SSP (resp. Algorithm~\ref{algo:HV} for the EMP) (line 3). We thus obtain the expected truncated sum $v$ (resp. the gain value $g$ and bias value $b$) of the strategy $\lambda$, for each block $B \in \partlump$. By definition of $\equivlump$, we have that for all $s, s' \in S$, if $s \equivlump s'$, then $v(s) = v(s')$ (resp. $g(s) = g(s')$ and $b(s) = b(s')$). Given a block $B \in \partlump$, we denote by $v(B)$ the unique expected truncated sum $v(s)$ (resp. $g(B)$ the unique gain value $g(s)$ and by $b(B)$ the unique bias value $b(s)$), for all $s \in B$.

\subsection{Improving strategies}

Given an MDP $\monMDP$ with cost function $\reward$ and the MC $\monMC$ induced by a strategy $\lambda$, we finally present a pseudo-antichain based algorithm to improve strategy $\lambda$ for the SSP, with the expected truncated sum $v$ obtained by solving the linear system (see line 8 of Algorithm~\ref{algo:symblicit}, and Algorithm~\ref{algo:SSPSI}). The improvement of a strategy for the EMP, with the gain $g$ or the bias $b$ values (see Algorithm~\ref{algo:HV}), is similar and is thus not detailed.

Recall that for all $s \in S$, we compute the set $\widehat{\enabledactions}$ of actions $\actionO \in \enabledactions$ that minimize the expression $\las = \reward(s, \actionO) + \sum_{s' \in S}\probmat(s, \actionO, s') \cdot v(s')$, and then we improve the strategy based on the computed $\widehat{\enabledactions}$. We give hereafter an approach based on pseudo-antichains which requires the next two steps. The first step consists in computing, for all $\actionO \in \ActionsO$, an equivalence relation $\equiva$ such that the value $\las$ is constant on each block of the relation. The second step uses the relations $\equiva$, with $\actionO \in \ActionsO$, to improve the strategy.

\paragraph{Computing value $l_{\actionO}$.}
Let $\actionO \in \ActionsO$ be a fixed action. 
We are looking for an equivalence relation $\equiva$ on the set $\enabledstates$ of states where action $\actionO$ is enabled, such that
\begin{equation*}
\forall s, s' \in \enabledstates : s \equiva s' \Rightarrow \las = \lasp.
\end{equation*}
Given $\equivlump$ the largest bisimulation for $\monMC$ and the induced partition $\partlump$, we have for each $s \in \enabledstates$
\begin{equation} \label{eq:las} \nonumber
\las = \reward(s, \actionO) + \sum_{C \in \partlump} \probmat(s, \actionO, C) \cdot v(C)
\end{equation}
since the value $v$ is constant on each block $C$. Therefore to get relation $\equiva$, it is enough to have $s \equiva s' \Rightarrow \reward(s, \actionO) = \reward(s', \actionO)$ and $\probmat(s, \actionO, C) = \probmat(s', \actionO, C), \forall C \in \partlump.$
We proceed by defining the following equivalence relations on $\enabledstates$. For the cost part, we use relation $\equivsigmarew$ defined in Section~\ref{subsec:symbrep}. For the probabilities part, for each block $C$ of $\partlump$, we define relation $\equivc$ such that $s \equivc s'$ iff $\probmat(s, \actionO, C) = \probmat(s', \actionO, C)$. The required relation $\equiva$ on $\enabledstates$ is then defined as the relation
$$\equiva ~=\ \ \ \equivsigmarew \cap \bigcap_{C \in \partlump}\!\!\equivc~=\ \ \ \equivsigmarew \cap \equivsigmaprob$$ 

Let us explain how to compute $\equiva$ with a pseudo-antichain based approach. Firstly, $\monMDP$ being $\ActionsI$-complete, the set $\enabledstates$ is obtained as $\enabledstates = \Pre_{\actionO,\actionI}(S)$ where $\actionI$ is an arbitrary action of $\ActionsI$.
Secondly, each relation $\equivc$ is the output obtained by a call to \textsc{Split}$(\enabledstates, C, \lambda)$ where $\lambda$ is defined on $\enabledstates$ by $\lambda(s) = \actionO$ for all $s \in \enabledstates$\footnote{As Algorithm~\textsc{Split} only works on $\enabledstates$, it is not a problem if $\lambda$ is not defined on $S \diff \enabledstates$.} (see Algorithm~\ref{algo:split}). Thirdly, we detail a way to compute $\equivsigmaprob$ from $\equivc$, for all $C \in \partlump$. Let $\partc = \{B_{C,1}, B_{C,2}, \dots, B_{C,k_C}\}$ be the partition of $\enabledstates$ induced by $\equivc$. For each $B_{C,i} \in \partc$, we denote by $\probmat(B_{C,i}, \actionO, C)$ the unique value $\probmat(s, \actionO, C)$, for all $s \in B_{C,i}$. Then, computing a block $D$ of $\equivsigmaprob$ consists in picking, for all $C \in \partlump$, one block $D_C$ among $B_{C,1}, B_{C,2}, \dots, B_{C,k_C}$, such that the intersection $D = \bigcap_{C \in \partlump} D_C$ is non empty. Recall that, by definition of the MDP $\monMDP$,  we have $\sum_{s' \in S}\probmat(s, \actionO, s') = 1$. Therefore, if $D$ is non empty, then $\sum_{C \in \partlump}\probmat(D_C, \actionO, C) = 1$. Finally, $\equiva$ is obtained as the intersection between $\equivsigmarew$ and  $\equivsigmaprob$.

Relation $\equiva$ induces a partition of $\enabledstates$ that we denote $\parta$. For each block $D \in\ \parta$, we denote by $\lad$ the unique value $\las$, for $s \in D$. 

\paragraph{Improving the strategy.} 

We now propose a pseudo-antichain based algorithm for improving strategy $\lambda$ by using relations $\equivlump$, $\equivstrat$, and $\equiva$, $\forall \actionO \in \ActionsO$ (see Algorithm~\ref{algo:improvestrat}).

We first compute for all $\actionO \in \ActionsO$, the equivalence relation $\equivalump\ =\ \equiva \cap \equivlump$ on $\enabledstates$.
Given $B \in \partalump$, we denote by $\laB$ the unique value $\las$ and by $v(B)$ the unique value $v(s)$, for all $s \in B$.  Let $\actionO \in \ActionsO$, we denote by $\partprom \subseteq \partalump$ the set of blocks $C$ for which the value $v(C)$ is improved by setting $\lambda(C) = \actionO$, that is 
$$\partprom = \{C \in \partalump \mid \lac < v(C)\}.$$ 
We then compute an ordered global list $\cal L$ made of the blocks of all sets $\partprom$, for all $\actionO \in \ActionsO$. It is ordered according to the decreasing value $\lac$. In this way, when traversing $\cal L$, we have more and more promising blocks to decrease $v$. 

From input $\cal L$ and $\equivstrat$, Algorithm~\ref{algo:improvestrat} outputs an equivalence relation $\equivstratp$ for a new strategy $\lambda'$ that improves $\lambda$.
Given $C \in \cal L$, suppose that $C$ comes from the relation $\equivalump$ ($\actionO$ is considered). Then for each $B \in S_{\equivstrat}$ such that $B \cap C  \neq \emptyset$ (line 4),  we improve the strategy by setting $\lambda'(B\cap C) = \actionO$, while the strategy $\lambda'$ is kept unchanged for $B\diff C$.
Algorithm~\ref{algo:improvestrat} outputs a partition $S_{\equivstratp}$ such that $s \equivstratp s' \Rightarrow \lambda'(s) = \lambda'(s')$ for the improved strategy $\lambda'$. If necessary, for efficiency reasons, we can compute a coarser relation for the new strategy $\lambda'$ by gathering blocks $B_1, B_2$ of $S_{\equivstratp}$, for all $B_1, B_2$ such that $\lambda'(B_1) = \lambda'(B_2)$. 

The correctness of Algorithm~\ref{algo:improvestrat} is due to the list $\cal L$, which is sorted according to the decreasing value $ \lac$. It ensures that the strategy is updated at each state $s$ to an action $\sigma \in \widehat{\enabledactions}$, i.e. an action $\actionO$ that minimizes the expression $\reward(s,\actionO) + \underset{s'\in S}{\sum}\probmat(s, \actionO, s')\cdot v_n(s')$ (cf. line $4$ of Algorithm~\ref{algo:SSPSI}).

\begin{algorithm}
\caption{\textsc{ImproveStrategy$({\cal L}, S_{\sim_{\lambda}})$}}
\begin{algorithmic}[1] \label{algo:improvestrat}
\FOR{$C \in {\cal L}$}
	\STATE $S_{\equivstratp} := \emptyset$
	\FOR{$B \in S_{\equivstrat}$}
		\IF{$B \cap C \neq \emptyset$}
			\STATE $S_{\equivstratp} := S_{\equivstratp} \cup \{B\cap C, B\diff C\}$
		\ELSE
			\STATE $S_{\equivstratp} := S_{\equivstratp} \cup B$
		\ENDIF
	\ENDFOR
	\STATE $S_{\equivstrat} := S_{\equivstratp}$
\ENDFOR
\STATE return $S_{\equivstratp}$
\end{algorithmic}
\end{algorithm}
\section{Experiments} \label{sec:experiments}
In this section, we present two application scenarios of the pseudo-antichain based symblicit algorithm of the previous section, one for the SSP problem and the other for the EMP problem. In both cases, we first show the reduction to monotonic MDPs that satisfy Assumptions~\ref{atwo} and~\ref{aone}, and we then present some experimental results. All our experiments have been done on a Linux platform with a $3.2$GHz CPU (Intel Core i7) and $12$GB of memory. Note that our implementations are single-threaded and thus use only one core. For all those experiments, the timeout is set to $10$ hours and is denoted by \timeout. Finally, we restrict the memory usage to $4$GB\footnote{Restricting the memory usage to $4$GB is enough to have a good picture of the behavior of each implementation with respect to the memory consumption.} and when an execution runs out of memory, we denote it by \memout.

\subsection{Stochastic shortest path on STRIPSs}
We consider the following application of the pseudo-antichain based symblicit algorithm for the SSP problem. In the field of planning, a class of problems called \textit{planning from STRIPSs}~\cite{fikes1972strips} operate with states represented by valuations of propositional variables. Informally, a STRIPS is defined by an \textit{initial} state representing the initial configuration of the system and a set of \textit{operators} that transform a state into another state. The problem of planning from STRIPSs then asks, given a valuation of propositional variables representing a set of \textit{goal} states, to find a sequence of operators that lead from the initial state to a goal one. Let us first formally define the notion of STRIPS and show that each STRIPS can be made monotonic. We will then add stochastic aspects and show how to construct a monotonic MDP from a monotonic stochastic STRIPS.

\paragraph{STRIPSs.} A \textit{STRIPS}~\cite{fikes1972strips} is a tuple $(P, I, (M, N), O)$ where $P$ is a finite set of \textit{conditions} (i.e. propositional variables), $I \subseteq P$ is a subset of conditions that are initially true (all others are assumed to be false), $(M, N)$, with $M, N \subseteq P$ and $M \cap N = \emptyset$, specifies which conditions are true and false, respectively, in order for a state to be considered a goal state, and $O$ is a finite set of \textit{operators}. An operator $o \in O$ is a pair $((\gamma, \theta), (\alpha, \delta))$ such that $(\gamma, \theta)$ is the \textit{guard} of $o$, that is, $\gamma \subseteq P$ (resp. $\theta \subseteq P$) is the set of conditions that must be true (resp. false) for $o$ to be executable, and $(\alpha, \delta)$ is the \textit{effect} of $o$, that is, $\alpha \subseteq P$ (resp. $\delta \subseteq P$) is the set of conditions that are made true (resp. false) by the execution of $o$. For all $((\gamma, \theta), (\alpha, \delta)) \in O$, we have that $\gamma \cap \theta = \emptyset$ and $\alpha \cap \delta = \emptyset$.

\noindent From a STRIPS, we derive a transition system as follows. The set of states is $2^P$, that is, a state is represented by the set of conditions that are true in it. The initial state is $I$. The set of goal states are states $Q$ such that $Q  \supseteq M$ and $Q \cap N = \emptyset$. There is a transition from state $Q$ to state $Q'$ under operator $o = ((\gamma, \theta), (\alpha, \delta))$ if $Q \supseteq \gamma$, $Q \cap \theta = \emptyset$ (the guard is satisfied) and $Q' = (Q \cap \alpha) \setminus \delta$ (the effect is applied). A standard problem is to ask whether or not there exists a path from the initial state to a goal state.

\paragraph{Monotonic STRIPSs.} 
A \textit{monotonic STRIPS (MS)} is a tuple $(P, I, M, O)$ where $P$ and $I$ are defined as for STRIPSs, $M \subseteq P$ specifies which conditions must be true in a goal state, and $O$ is a finite set of operators. In the MS definition, an operator $o \in O$ is a pair $(\gamma, (\alpha, \delta))$ where $\gamma \subseteq P$ is the guard of $o$, that is, the set of conditions that must be true for $o$ to be executable, and $(\alpha, \delta)$ is the effect of $o$ as in the STRIPS definition. 
MSs thus differ from STRIPS in the sense that guards only apply on conditions that are true in states, and goal states are only specified by true conditions. The monotonicity will appear more clearly when we will derive hereafter monotonic MDPs from MSs.

\noindent
Each STRIPS ${\cal S} = (P, I, (M, N), O)$ can be made monotonic by duplicating the set of conditions, in the following way. We denote by $\overline{P}$ the set $\{\overline{p} \mid p \in P\}$ containing a new condition $\overline{p}$ for each $p \in P$ such that $\overline{p}$ represents the negation of the propositional variable $p$. We construct from $\cal S$ an MS ${\cal S}' = (P', I', M', O')$ such that $P' = P \cup \overline{P}$, $I' = I \cup \overline{P\diff I}\subseteq P'$, $M' = M \cup \overline{N} \subseteq P'$ and $O' = \{(\gamma \cup \overline{\theta}, (\alpha \cup \overline{\delta}, \delta \cup \overline{\alpha})) \mid ((\gamma, \theta), (\alpha, \delta)) \in O\}$. It is easy to check that ${\cal S}$ and ${\cal S}'$ are equivalent (a state $Q$ in ${\cal S}$ has its counterpart $Q \cup \overline{P\diff Q}$ in ${\cal S}'$). In the following, we thus only consider MSs. 

\begin{example} \label{ex:ST}
To illustrate the notion of MS, let us consider the following example of the monkey trying to reach a bunch of bananas (cf. Example~\ref{ex:monkey}). Let $(P, I, M, O)$ be an MS such that $P =$ \{\textit{box}, \textit{stick}, \textit{bananas}$\}$, $I = \emptyset$, $M = \{\textit{bananas}\}$, and $O = \{\textit{takebox}, \textit{takestick}, \textit{takebananas}\}$ where $\textit{takebox} = (\emptyset, (\{\textit{box}\}, \emptyset))$,  $\textit{takestick} = (\emptyset, (\{\textit{stick}\}, \emptyset))$ and $\textit{takebananas} = (\{\textit{box},\textit{stick}\}, (\{\textit{bananas}\}, \emptyset))$. In this MS, a condition $p \in P$ is true when the monkey possesses the item corresponding to $p$. At the beginning, the monkey possesses no item, i.e. $I$ is the empty set, and its goal is to get the \textit{bananas}, i.e. to reach a state $s \supseteq \{\textit{bananas}\}$. This can be done by first executing the operators \textit{takebox} and \textit{takestick} to respectively get the \textit{box} and the \textit{stick}, and then executing \textit{takebananas}, whose guard is $\{\textit{box}, \textit{stick}\}$.
\end{example}

\paragraph{Monotonic stochastic STRIPSs.} MSs can be extended with stochastic aspects as follows~\cite{blum2000probabilistic}. Each operator $o = (\gamma, \pi) \in O$ now consists of a guard $\gamma$ as before, and an effect given as a probability distribution $\pi : 2^P\times 2^P \rightarrow [0,1]$ on the set of pairs $(\alpha, \delta)$. An MS extended with such stochastic aspects is called a \textit{monotonic stochastic STRIPS (MSS)}. 

\noindent
Additionally, we associate with an MSS $(P, I, M, O)$ a cost function $C : O \rightarrow \R_{> 0}$ that associates a strictly positive cost with each operator. The problem of planning from MSSs is then to minimize the expected truncated sum up to the set of goal states from the initial state, i.e. this is a version of the SSP problem.

\begin{example}\label{ex:mst}
We extend the MS of Example~\ref{ex:ST} with stochastic aspects to illustrate the notion of MSS. Let $(P, I, M, O)$ be an MSS such that $P$, $I$ and $M$ are defined as in Example~\ref{ex:ST}, and $O =$ \{\textit{takebox}, \textit{takestick}, \textit{takebananaswithbox}, \textit{takebananaswithstick}, \textit{takebananaswithboth}\} where 
\begin{itemize}
\item $\textit{takebox} = (\emptyset, (1 : (\{\textit{box}\}, \emptyset)))$,
\item $\textit{takestick} = (\emptyset, (1 : (\{\textit{stick}\}, \emptyset)))$,
\item $\textit{takebananaswithbox} = (\{\textit{box}\}, (\frac{1}{4}: (\{\textit{bananas}\}, \emptyset), \frac{3}{4}: (\emptyset, \emptyset)))$,
\item $\textit{takebananaswithstick} = (\{\textit{stick}\}, (\frac{1}{5}: (\{\textit{bananas}\}, \emptyset), \frac{4}{5}: (\emptyset, \emptyset)))$, and 
\item $\textit{takebananaswithboth} = (\{\textit{box}, \textit{stick}\}, (\frac{1}{2}: (\{\textit{bananas}\}, \emptyset), \frac{1}{2}: (\emptyset, \emptyset)))$.
\end{itemize}
In this MSS, the monkey has a strictly positive probability to fail reaching the \textit{bananas}, whatever the items it uses. However, the probability of success increases when it has both the \textit{box} and the \textit{stick}.
\end{example}

In the following, we show that MSSs naturally define monotonic MDPs on which the pseudo-antichain based symblicit algorithm of Section~\ref{sec:paalgo} can be applied.

\paragraph{From MSSs to monotonic MDPs.} Let ${\cal S} = (P, I, M, O)$ be an MSS. We can derive from ${\cal S}$ an MDP $M_{\cal S} = (S, \ActionsO, \ActionsI, \edges, \distr)$ together with a set of goal states $G$ and a cost function $\reward$ such that:
\begin{itemize}
\itemsep0.1em
\item $S = 2^P$,
\item $G = \{s \in S \mid s \supseteq M\}$,
\item $\ActionsO = O$, and for all $s \in S$, $\enabledactions = \{(\gamma, \pi) \in \ActionsO \mid s \supseteq \gamma\}$,
\item $\ActionsI = \{(\alpha, \delta) \in 2^P\times2^P \mid \exists (\gamma, \pi) \in O, (\alpha, \delta) \in \support(\pi)\}$,
\item $\edges$, $\distr$ and $\reward$ are defined for all $s \in S$ and $\actionO = (\gamma, \pi) \in \enabledactions$, such that:
\begin{itemize} 
\item for all $\actionI = (\alpha, \delta) \in \ActionsI$, $\edges(s, \actionO)(\actionI) = (s \cup \alpha) \diff \delta$,
\item for all $\actionI \in \ActionsI$, $\distr(s, \actionO)(\actionI) = \pi(\actionI)$, and
\item $\reward(s, \actionO) = C(\actionO)$.
\end{itemize}
\end{itemize}
Note that we might have that $M_{\cal S}$ is not $\Sigma$-non-blocking, if no operator can be applied on some state of $\cal S$. In this case, we get a $\Sigma$-non-blocking MDP from $M_{\cal S}$ by eliminating states $s$ with $\enabledactions = \emptyset$  as long as it is necessary.

\begin{lemma} \label{lem:StripsMono}
The MDP $M_{\cal S}$ is monotonic, $G$ is closed, and functions $\distr, \reward$ are independent from $S$.
\end{lemma}
\begin{proof}
First, $S$ is equipped with the partial order $\supseteq$ and $(S, \supseteq)$ is a semilattice. Second, $S$ is closed for $\supseteq$ by definition. 
Thirdly,  we have that $\supseteq$ is compatible with $\edges$. Indeed, for all $s, s' \in S$ such that $s \supseteq s'$, for all $\actionO \in \ActionsO$ and $\actionI = (\alpha, \delta) \in \ActionsI$, $\edges(s, \actionO)(\actionI) = (s \cup \alpha)\diff \delta \supseteq (s' \cup \alpha) \diff \delta = \edges(s', \actionO)(\actionI)$. Finally the set $G =\ \antclos\{M\}$ of goal states is closed for $\supseteq$, and $\distr, \reward$ are clearly independent from $S$.
\qed\end{proof}

\paragraph{Symblicit algorithm.} In order to apply the pseudo-antichain based symblicit algorithm of Section~\ref{sec:paalgo} on the monotonic MDPs derived from MSSs, Assumptions~\ref{atwo} and~\ref{aone} must hold. Let us show that Assumption~\ref{aone} is satisfied. For all $s \in S$, $\actionO = (\gamma, \pi) \in \enabledactions$ and $\actionI = (\alpha, \delta) \in \ActionsI$, we clearly have an algorithm for computing $\edges(s, \actionO)(\actionI) = (s \cup \alpha)\diff \delta$, and $\distr(s, \actionO)(\actionI) = \pi(\actionI)$. Let us now consider Assumption~\ref{atwo}. An algorithm for computing $\lceil \Pre_{\actionO,\actionI}(\antclos\{x\})\rceil$, for all $x \in S$, $\actionO \in \ActionsO$ and $\actionI \in \ActionsI$, is given by the next proposition. 
\begin{proposition}
Let $x \in S$, $\actionO = (\gamma, \pi) \in \ActionsO$ and $\actionI = (\alpha, \delta) \in \ActionsI$. If $x \cap \delta \neq \emptyset$, then  $\lceil \Pre_{\actionO,\actionI}(\antclos\{x\})\rceil = \emptyset$, otherwise $\lceil \Pre_{\actionO,\actionI}(\antclos\{x\})\rceil = \{\gamma \cup (x \diff \alpha)\}$.
\end{proposition}
\begin{proof}
Suppose first that $x \cap \delta = \emptyset$.

We first prove that $s = \gamma \cup (x \diff \alpha) \in \Pre_{\actionO,\actionI}(\antclos\{x\})$.  We have to show that $\actionO \in \Sigma_{s}$ and  $\edges(s, \actionO)(\actionI) \in \antclos\{x\}$. Recall that $\actionO = (\gamma, \pi)$. We have that $s = \gamma \cup (x\diff \alpha) \supseteq \gamma$, showing that $\actionO \in \Sigma_{s}$. We have that $\edges(s, \actionO)(\actionI)  = (\gamma \cup (x \diff \alpha) \cup \alpha) \diff \delta =  (\gamma \cup x \cup \alpha) \diff \delta \supseteq x$ since $x \cap \delta = \emptyset$. We thus have that $\edges(s, \actionO)(\actionI) \in \antclos\{x\}$.

We then prove that for all $s \in \Pre_{\actionO,\actionI}(\antclos\{x\})$, $s \in \antclos\{\gamma \cup (x \diff \alpha)\}$, i.e. $s \supseteq \gamma \cup (x \diff \alpha)$. Let $s \in \Pre_{\actionO,\actionI}(\antclos\{x\})$. We have that $\actionO \in \Sigma_{s}$ and $\edges(s, \actionO)(\actionI) \in \antclos\{x\}$, that is,  $s \supseteq \gamma$ and $\edges(s, \actionO)(\actionI) = (s \cup \alpha) \diff \delta \supseteq x$. By classical set properties, it follows that $(s \cup \alpha) \supseteq x$, and then $s  \supseteq x\diff\alpha$. Finally, since $s \supseteq \gamma$, we have $s \supseteq \gamma \cup (x \diff \alpha)$, as required.

Suppose now that $x \cap \delta \neq \emptyset$, then $\Pre_{\actionO,\actionI}(\antclos\{x\}) = \emptyset$. Indeed for all $s \in \antclos\{x\}$, we have  $s \cap \delta \neq \emptyset$, and by definition of $\edges$, there is no $s'$ such that $\edges(s', \actionO)(\actionI) = s$.
\qed\end{proof}

Finally, notice that for the class of monotonic MDPs derived from MSSs, the symbolic representations described in Section~\ref{subsec:symbrep} are compact, since $G$ is closed and  $\distr, \reward$ are independent from $S$ (see Lemma~\ref{lem:StripsMono}).
Therefore we have all the required ingredients for an efficient pseudo-antichain based algorithm to solve the SSP problem for MSSs. The next experiments show its performance.

\paragraph{Experiments.} We have implemented  in Python and C the pseudo-antichain based symblicit algorithm for the SSP problem. The C language is used for all the low level operations while the orchestration is done with Python. The binding between C and Python is realized with the ctypes library of Python. The source code is publicly available at \url{http://lit2.ulb.ac.be/STRIPSSolver/}, together with the two benchmarks presented in this section. We compared our implementation with the purely explicit strategy iteration algorithm implemented in the development release 4.1.dev.r7712 of the tool $\sf{PRISM}$~\cite{KNP11}, since to the best of our knowledge, there is no tool implementing an MTBDD based symblicit algorithm for the SSP problem.\footnote{A comparison with an MTBDD based symblicit algorithm is done in the second application for the EMP problem.} Note that this explicit implementation exists primarily to prototype new techniques and is thus not fully optimized~\cite{Parker}. Note that value iteration algorithms are also implemented in $\sf{PRISM}$. 
While those algorithms are usually efficient, they only compute approximations. As a consequence, for the sake of a fair comparison, we consider here only the performances of strategy iteration algorithms. 

\medskip
The first benchmark ($\sf{Monkey}$) is obtained from Example~\ref{ex:mst}. In this benchmark, the monkey has several items at its disposal to reach the bunch of bananas, one of them being a stick. However, the stick is available as a set of several pieces that the monkey has to assemble. Moreover, the monkey has multiple ways to build the stick as there are several sets of pieces that can be put together. However, the time required to build a stick varies from a set of pieces to another. Additionally, we add useless items in the room: there is always a set of pieces from which the probability of getting a stick is $0$. The operators of getting some items are stochastic, as well as the operator of getting the bananas: the probability of success varies according to the owned items (cf. Example~\ref{ex:mst}). The benchmark is parameterized in the number $p$ of pieces required to build a stick, and in the number $s$ of sticks that can be built. Note that the monkey can only use one stick, and thus has no interest to build a second stick if it already has one. Results are given in Table~\ref{table:STRIPS1}. 

\begin{table}[h!]
	\caption{Stochastic shortest path on the $\sf{Monkey}$ benchmark. The column $(s, p)$ gives the parameters of the problem, $\ETP_\lambda$ the expected truncated sum of the computed strategy $\lambda$, and $|M_{\cal S}|$ the number of states of the MDP.
For the pseudo-antichain based implementation ($\sf{PA}$), $\#$it is the number of iterations of the strategy iteration algorithm, $|\partlump|$ the maximum size of computed bisimulation quotients, \textit{lump} the total time spent for lumping, \textit{syst} the total time spent for solving the linear systems, and \textit{impr} the total time spent for improving the strategies. For the explicit implementation ($\sf{Explicit}$), \textit{constr} is the time spent for model construction and \textit{strat} the time spent for the strategy iteration algorithm. For both implementations, \textit{total} is the total execution time and \textit{mem} the total memory consumption. All times are given in seconds and all memory consumptions are given in megabytes.}
	\label{table:STRIPS1}
	\centering
		\scriptsize
 		\begin{tabular}{|r|r|r||r|r|r|r|r|r|r||r|r|r|r|r|r|r|r|r|r|r|r|r|r|r|}
		\hline
	  	& & & \multicolumn{7}{|c||}{{\small $\sf{PA}$}} & \multicolumn{4}{|c|}{{\small $\sf{Explicit}$}}\rule[-2pt]{0pt}{10pt}\\
		$\ (s, p) \ $ & $\ETP_\lambda$  & $|M_{\cal S}|$  & \ $\#$it \  & $\ |\partlump|\ $ & \ lump  \ &  \ syst \  &  \ impr  \ & \  total \  & mem & \ constr \  & \ strat \ & \ total \ & \ mem\ \rule[-3pt]{0pt}{10pt}\\
\hline\hline
$\ (1,2)\ $ & $\ 35.75\ $ & $\ 256\ $ & $\ 4\ $ & $\ 15\ $ & $\ 0.01\ $ & $\ 0.00\ $ & $\ 0.02\ $  & $\ 0.03\ $ & $\ 15.6\ $ & $\ 0.4\ $ & $\ 0.03\ $  & $\ 0.43\ $ & $\ 178.2\ $\rule[-3pt]{0pt}{10pt}\\
$\ (1,3)\ $ & $\ 35.75\ $ & $\ 1024\ $ & $\ 5\ $ & $\ 19\ $ & $\ 0.04\ $ & $\ 0.00\ $ & $\ 0.03\ $  & $\ 0.07\ $ & $\ 15.8\ $ & $\ 3.42\ $ & $\ 0.09\ $  & $\ 3.51\ $ & $\ 336.7\ $\rule[-3pt]{0pt}{10pt}\\
$\ (1,4)\ $ & $\ 35.75\ $ & $\ 4096\ $ & $\ 6\ $ & $\ 31\ $ & $\ 0.17\ $ & $\ 0.00\ $ & $\ 0.12\ $  & $\ 0.29\ $ & $\ 16.3\ $ & $\ 55.29\ $ & $\ 0.16\ $  & $\ 55.45\ $ & $\ 1735.1\ $\rule[-3pt]{0pt}{10pt}\\
$\ (1,5)\ $ & $\ 36.00\ $ & $\ 16384\ $ & $\ 7\ $ & $\ 39\ $ & $\ 0.75\ $ & $\ 0.00\ $ & $\ 0.62\ $  & $\ 1.37\ $ & $\ 17.1\ $ & $\ \ $ & $\ \ $  & $\ \ $ & $\ \memout\ $\rule[-3pt]{0pt}{10pt}\\
\hline
$\ (2,2)\ $ & $\ 34.75\ $ & $\ 1024\ $ & $\ 5\ $ & $\ 19\ $ & $\ 0.05\ $ & $\ 0.00\ $ & $\ 0.03\ $  & $\ 0.09\ $ & $\ 15.8\ $ & $\ 3.2\ $ & $\ 0.12\ $  & $\ 3.32\ $ & $\ 379.9\ $\rule[-3pt]{0pt}{10pt}\\
$\ (2,3)\ $ & $\ 34.75\ $ & $\ 8192\ $ & $\ 5\ $ & $\ 37\ $ & $\ 0.32\ $ & $\ 0.00\ $ & $\ 0.13\ $  & $\ 0.45\ $ & $\ 16.4\ $ & $\ 240.66\ $ & $\ 0.30\ $  & $\ 240.96\ $ & $\ 3463.2\ $\rule[-3pt]{0pt}{10pt}\\
$\ (2,4)\ $ & $\ 34.75\ $ & $\ 65536\ $ & $\ 6\ $ & $\ 45\ $ & $\ 2.39\ $ & $\ 0.01\ $ & $\ 1.04\ $  & $\ 3.44\ $ & $\ 18.0\ $ & $\ \ $ & $\ \ $  & $\ \ $ & $\ \memout\ $\rule[-3pt]{0pt}{10pt}\\
$\ (2,5)\ $ & $\ 35.75\ $ & $\ 524288\ $ & $\ 7\ $ & $\ 65\ $ & $\ 27.56\ $ & $\ 0.02\ $ & $\ 10.13\ $  & $\ 37.71\ $ & $\ 23.4\ $ & $\ \ $ & $\ \ $  & $\ \ $ & $\ \memout\ $\rule[-3pt]{0pt}{10pt}\\
\hline
$\ (3,2)\ $ & $\ 35.75\ $ & $\ 4096\ $ & $\ 4\ $ & $\ 23\ $ & $\ 0.09\ $ & $\ 0.00\ $ & $\ 0.07\ $  & $\ 0.16\ $ & $\ 16.0\ $ & $\ 60.43\ $ & $\ 0.16\ $  & $\ 60.59\ $ & $\ 1625.8\ $\rule[-3pt]{0pt}{10pt}\\
$\ (3,3)\ $ & $\ 35.75\ $ & $\ 65536\ $ & $\ 5\ $ & $\ 43\ $ & $\ 1.14\ $ & $\ 0.00\ $ & $\ 0.43\ $  & $\ 1.57\ $ & $\ 17.3\ $ & $\ \ $ & $\ \ $  & $\ \ $ & $\ \memout\ $\rule[-3pt]{0pt}{10pt}\\
$\ (3,4)\ $ & $\ 35.75\ $ & $\ 1048576\ $ & $\ 6\ $ & $\ 57\ $ & $\ 12.89\ $ & $\ 0.01\ $ & $\ 4.92\ $  & $\ 17.83\ $ & $\ 21.7\ $ & $\ \ $ & $\ \ $  & $\ \ $ & $\ \memout\ $\rule[-3pt]{0pt}{10pt}\\
$\ (3,5)\ $ & $\ 36.00\ $ & $\ 16777216\ $ & $\ 7\ $ & $\ 88\ $ & $\ 208.33\ $ & $\ 0.05\ $ & $\ 63.73\ $  & $\ 272.13\ $ & $\ 37.5\ $ & $\ \ $ & $\ \ $  & $\ \ $ & $\ \memout\ $\rule[-3pt]{0pt}{10pt}\\
\hline
$\ (4,2)\ $ & $\ 35.75\ $ & $\ 16384\ $ & $\ 4\ $ & $\ 29\ $ & $\ 0.22\ $ & $\ 0.02\ $ & $\ 0.14\ $  & $\ 0.38\ $ & $\ 16.3\ $ & $\ 1114.19\ $ & $\ 0.70\ $  & $\ 1114.89\ $ & $\ 1704.3\ $\rule[-3pt]{0pt}{10pt}\\
$\ (4,3)\ $ & $\ 35.75\ $ & $\ 524288\ $ & $\ 5\ $ & $\ 50\ $ & $\ 2.72\ $ & $\ 0.00\ $ & $\ 1.26\ $  & $\ 4.00\ $ & $\ 18.3\ $ & $\ \ $ & $\ \ $  & $\ \ $ & $\ \memout\ $\rule[-3pt]{0pt}{10pt}\\
$\ (4,4)\ $ & $\ 35.75\ $ & $\ 16777216\ $ & $\ 6\ $ & $\ 87\ $ & $\ 45.68\ $ & $\ 0.04\ $ & $\ 22.41\ $  & $\ 68.14\ $ & $\ 25.0\ $ & $\ \ $ & $\ \ $  & $\ \ $ & $\ \memout\ $\rule[-3pt]{0pt}{10pt}\\
$\ (4,5)\ $ & $\ 36.00\ $ & $\ 536870912\ $ & $\ 7\ $ & $\ 114\ $ & $\ 724.77\ $ & $\ 0.11\ $ & $\ 532.46\ $  & $\ 1257.41\ $ & $\ 60.9\ $ & $\ \ $ & $\ \ $  & $\ \ $ & $\ \memout\ $\rule[-3pt]{0pt}{10pt}\\
\hline
$\ (5,2)\ $ & $\ 35.75\ $ & $\ 65536\ $ & $\ 4\ $ & $\ 31\ $ & $\ 0.36\ $ & $\ 0.00\ $ & $\ 0.18\ $  & $\ 0.54\ $ & $\ 16.6\ $ & $\ 20312.67\ $ & $\ 3.50\ $  & $\ 20316.17\ $ & $\ 2342.6\ $\rule[-3pt]{0pt}{10pt}\\
$\ (5,3)\ $ & $\ 35.75\ $ & $\ 4194304\ $ & $\ 5\ $ & $\ 56\ $ & $\ 5.71\ $ & $\ 0.02\ $ & $\ 2.47\ $  & $\ 8.20\ $ & $\ 19.5\ $ & $\ \ $ & $\ \ $  & $\ \ $ & $\ \memout\ $\rule[-3pt]{0pt}{10pt}\\
$\ (5,4)\ $ & $\ 35.75\ $ & $\ 268435456\ $ & $\ 6\ $ & $\ 97\ $ & $\ 95.49\ $ & $\ 0.04\ $ & $\ 101.27\ $  & $\ 196.83\ $ & $\ 31.3\ $ & $\ \ $ & $\ \ $  & $\ \ $ & $\ \memout\ $\rule[-3pt]{0pt}{10pt}\\
$\ (5,5)\ $ & $\ 36.00\ $ & $\ 17179869184\ $ & $\ 7\ $ & $\ 152\ $ & $\ 1813.78\ $ & $\ 0.08\ $ & $\ 5284.31\ $  & $\ 7098.40\ $ & $\ 81.3\ $ & $\ \ $ & $\ \ $  & $\ \ $ & $\ \memout\ $\rule[-3pt]{0pt}{10pt}\\
\hline
		\end{tabular}
		\normalsize
\end{table}

The second benchmark ($\sf{Moats~and~castles}$) is an adaptation of a benchmark of~\cite{DBLP:conf/aips/MajercikL98} as proposed in~\cite{blum2000probabilistic}\footnote{In~\cite{blum2000probabilistic}, the authors study a different problem that is to maximize the probability of reaching the goal within a given number of steps.}. The goal is to build a sand castle on the beach; a moat can be dug before in a way to protect it. We consider up to $7$ discrete depths of moat. The operator of building the castle is stochastic: there is a strictly positive probability for the castle to be demolished by the waves. However, the deeper the moat is, the higher the probability of success is. For example, the first depth of moat offers a probability $\frac{1}{4}$ of success, while with the second depth of moat, the castle has probability $\frac{9}{20}$ to resist to the waves. The optimal strategy for this problem is to dig up to a given depth of moat and then repeat the action of building the castle until it succeeds. The optimal depth of moat then depends on the cost of the operators and the respective probability of successfully building the castle for each depth of moat. To increase the difficulty of the problem, we consider building several castles, each one having its own moat. The benchmark is parameterized in the number $d$ of depths of moat that can be dug, and the number $c$ of castles that have to be built. Results are given in Table~\ref{table:STRIPS2}.

\begin{table}[h!]
	\caption{Stochastic shortest path on the $\sf{Moats~and~castles}$ benchmark. The column $(c, d)$ gives the parameters of the problem and all other columns have the same meaning as in Table~\ref{table:STRIPS1}.}
	\label{table:STRIPS2}
	\centering
	\scriptsize
 		\begin{tabular}{|r|r|r||r|r|r|r|r|r|r||r|r|r|r|r|r|r|r|r|r|r|r|r|r|r|}
		\hline
	  	& & & \multicolumn{7}{|c||}{{\small $\sf{PA}$}} & \multicolumn{4}{|c|}{{\small $\sf{Explicit}$}}\rule[-2pt]{0pt}{10pt}\\
		$\ (c, d) \ $ & $\ETP_\lambda$  & $|M_{\cal S}|$  & \ $\#$it \  & $\ |\partlump|\ $ & \ lump  \ &  \ syst \  &  \ impr  \ & \  total \  & mem & \ constr \  & \ strat \ & \ total \ & \ mem\ \rule[-3pt]{0pt}{10pt}\\
\hline\hline
$\ (2,3)\ $ & $\ 39.3333\ $ & $\ 256\ $ & $\ 3\ $ & $\ 17\ $ & $\ 0.05\ $ & $\ 0.01\ $ & $\ 0.04\ $  & $\ 0.10\ $ & $\ 15.8\ $ & $\ 0.46\ $ & $\ 0.03\ $  & $\ 0.49\ $ & $\ 206.9\ $\rule[-3pt]{0pt}{10pt}\\
$\ (2,4)\ $ & $\ 34.6667\ $ & $\ 1024\ $ & $\ 3\ $ & $\ 34\ $ & $\ 0.41\ $ & $\ 0.00\ $ & $\ 0.17\ $  & $\ 0.58\ $ & $\ 16.5\ $ & $\ 6.30\ $ & $\ 0.08\ $  & $\ 6.38\ $ & $\ 483.8\ $\rule[-3pt]{0pt}{10pt}\\
$\ (2,5)\ $ & $\ 32.2222\ $ & $\ 4096\ $ & $\ 3\ $ & $\ 49\ $ & $\ 1.36\ $ & $\ 0.00\ $ & $\ 0.45\ $  & $\ 1.82\ $ & $\ 17.3\ $ & $\ 133.46\ $ & $\ 0.20\ $  & $\ 133.66\ $ & $\ 1202.5\ $\rule[-3pt]{0pt}{10pt}\\
$\ (2,6)\ $ & $\ 32.2222\ $ & $\ 16384\ $ & $\ 3\ $ & $\ 66\ $ & $\ 9.71\ $ & $\ 0.01\ $ & $\ 1.95\ $  & $\ 11.68\ $ & $\ 19.3\ $ & $\ 2966.01\ $ & $\ 0.79\ $  & $\ 2966.80\ $ & $\  1706.2\ $\rule[-3pt]{0pt}{10pt}\\
\hline
$\ (3,2)\ $ & $\ 72.6667\ $ & $\ 512\ $ & $\ 3\ $ & $\ 45\ $ & $\ 0.52\ $ & $\ 0.00\ $ & $\ 0.24\ $  & $\ 0.77\ $ & $\ 16.6\ $ & $\ 1.77\ $ & $\ 0.06\ $  & $\ 1.83\ $ & $\ 282.9\ $\rule[-3pt]{0pt}{10pt}\\
$\ (3,3)\ $ & $\ 59.0000\ $ & $\ 4096\ $ & $\ 3\ $ & $\ 84\ $ & $\ 12.58\ $ & $\ 0.03\ $ & $\ 2.73\ $  & $\ 15.35\ $ & $\ 20.2\ $ & $\ 149.44\ $ & $\ 0.20\ $  & $\ 149.64\ $ & $\ 1205.5\ $\rule[-3pt]{0pt}{10pt}\\
$\ (3,4)\ $ & $\ 52.0000\ $ & $\ 32768\ $ & $\ 3\ $ & $\ 219\ $ & $\ 129.17\ $ & $\ 0.05\ $ & $\ 21.56\ $  & $\ 150.83\ $ & $\ 30.7\ $ & $\ 14658.22\ $ & $\ 2.47\ $  & $\ 14660.69\ $ & $\ 1610.9\ $\rule[-3pt]{0pt}{10pt}\\
$\ (3,5)\ $ & $\ 48.3333\ $ & $\ 262144\ $ & $\ 3\ $ & $\ 357\ $ & $\ 658.86\ $ & $\ 0.13\ $ & $\ 81.08\ $  & $\ 740.17\ $ & $\ 49.1\ $ & $\ \ $ & $\ \ $  & $\ \ $ & $\ \memout\ $\rule[-3pt]{0pt}{10pt}\\
$\ (3,6)\ $ & $\ 48.3333\ $ & $\ 2097152\ $ & $\ 3\ $ & $\ 595\ $ & $\ 10730.09\ $ & $\ 0.42\ $ & $\ 865.48\ $  & $\ 11596.71\ $ & $\ 145.8\ $ & $\ \ $ & $\ \ $  & $\ \ $ & $\ \memout\ $\rule[-3pt]{0pt}{10pt}\\
\hline
$\ (4,2)\ $ & $\ 96.8889\ $ & $\ 4096\ $ & $\ 3\ $ & $\ 132\ $ & $\ 31.61\ $ & $\ 0.03\ $ & $\ 12.06\ $  & $\ 43.72\ $ & $\ 26.5\ $ & $\ 173.40\ $ & $\ 0.22\ $  & $\ 173.62\ $ & $\ 1211.2\ $\rule[-3pt]{0pt}{10pt}\\
$\ (4,3)\ $ & $\ 78.6667\ $ & $\ 65536\ $ & $\ 3\ $ & $\ 464\ $ & $\ 1376.94\ $ & $\ 0.21\ $ & $\ 217.06\ $  & $\ 1594.48\ $ & $\ 82.2\ $ & $\ \ $ & $\ \ $  & $\ \ $ & $\ \memout\ $\rule[-3pt]{0pt}{10pt}\\
\hline
		\end{tabular}
	\normalsize
\end{table}

On those two benchmarks, we observe that the explicit implementation quickly runs out of memory when the state space of the MDP grows. Indeed, with this method, we were not able to solve MDPs with more than $65536$ (resp. $32768$) states in Table~\ref{table:STRIPS1} (resp. Table~\ref{table:STRIPS2}). On the other hand, the symblicit algorithm behaves well on large models: the memory consumption never exceeds $150$Mo and this even for MDPs with hundreds of millions of states. For instance, the example $(5, 5)$ of the $\sf{Monkey}$ benchmark is an MDP of more than $17$ billions of states that is solved in less than $2$ hours with only $82$Mo of memory\footnote{On our benchmarks, the value iteration algorithm of ${\sf PRISM}$~performs better than the strategy iteration one w.r.t. the run time and memory consumption. However, it still consumes more memory than the pseudo-antichain based algorithm, and runs out of memory on several examples.}.

\subsection{Expected mean-payoff with \LTLMP synthesis}
We consider another application of the pseudo-antichain based symblicit algorithm, but now for the EMP problem. This application is related to the problems of \LTLMP realizability and synthesis~\cite{DBLP:journals/corr/abs-1210-3539,DBLP:conf/tacas/BohyBFR13}. Let us fix some notations and definitions. Let $\phi$ be an \LTL formula defined over the set $P = I \uplus O$ of signals and let $\Sigma_P = 2^P$, $\Sigma_O = 2^O$ and $\Sigma_I = 2^I$. Let $\Lit(O) = \{o \mid o \in O\} \cup \{\neg o \mid o \in O\}$ be the set of literals over $O$. Let $w : \Lit(O) \mapsto \Z$ be a weight function where positive numbers represent rewards.\footnote{Note that in~\cite{DBLP:journals/corr/abs-1210-3539,DBLP:conf/tacas/BohyBFR13}, the weight function $w$ is more general since it also associates values to $\Lit(I)$. However, for this application, we restrict $w$ to $\Lit(O)$.} This function is extended to $\Sigma_O$ as follows:  $w(\sigma) = \Sigma_{o \in \sigma} w(o) +  \Sigma_{o \in O \setminus \{\sigma\}} w(\neg o)$ for all $\sigma \in \Sigma_O$.

\paragraph{$\LTLMP$ realizability and synthesis.} The problem of \LTLMP realizability is best seen as a game between two players, Player $O$ and Player $I$. This game is infinite and such that at each turn $k$, Player $O$ gives a subset $o_k \in \Sigma_O$ and Player $I$ responds by giving a subset $i_k \in \Sigma_I$. The outcome of the game is the infinite word $(o_0 \cup i_0)(o_1 \cup i_1)\dots \in \Sigma_P^\omega$. A strategy for Player $O$ is a mapping $\lambda_O : (\Sigma_O\Sigma_I)^* \rightarrow \Sigma_O$, while a strategy for Player $I$ is a mapping $\lambda_I : (\Sigma_O\Sigma_I)^*\Sigma_O \rightarrow \Sigma_I$. The outcome of the strategies $\lambda_O$ and $\lambda_I$ is the word $\Outcome(\lambda_O, \lambda_I) = (o_0 \cup i_0)(o_1 \cup i_1)\dots$ such that $o_0 = \lambda_O(\epsilon), i_0 = \lambda_I(o_0)$ and for all $k \geq 1, o_k = \lambda_O(o_0i_0\dots o_{k-1}i_{k-1})$ and $i_k = \lambda_I(o_0i_0\dots o_{k-1}i_{k-1}o_k)$. A \textit{value} $\mpval(u)$ is associated with each outcome $u \in \Sigma_P^\omega$ such that
\begin{center}
$\mpval(u) =$
$\begin{cases} \liminf_{n \rightarrow \infty} \frac{1}{n} \sum_{k=0}^{n-1} w(o_k) \text{ if } u \models \phi\\  - \infty \text{ otherwise} \end{cases}$
\end{center}
i.e. $\mpval(u)$ is the mean-payoff value of $u$ if $u$ satisfies $\phi$, otherwise, it is $-\infty$. Given an $\LTL$ formula $\phi$ over $P$, a weight function $w$ and a threshold value $\nu \in \Z$, the $\LTLMP$ \textit{realizability problem} asks to decide whether there exists a strategy $\lambda_O$ for Player $O$ such that $\mpval(\Outcome(\lambda_O, \lambda_I)) \geq \nu$ for all strategies $\lambda_I$ of Player~$I$. If the answer is $\mathsf{Yes}$, $\phi$ is said $\MP$-realizable. The $\LTLMP$ \textit{synthesis problem} is then to produce such a strategy $\lambda_O$ for Player $O$. 

To illustrate the problems of $\LTLMP$ realizability and synthesis, let us consider the following specification of a server that should grant exclusive access to a resource to two clients.

\begin{example} \label{ex:ltl} 
A client requests access to the resource by setting to true its request signal ($r_1$ for client~$1$ and $r_2$ for client~$2$), and the server grants those requests by setting to true the respective grant signal $g_1$ or $g_2$. We want to synthesize a server that eventually grants any client request, and that only grants one request at a time. Additionally, we ask client~$2$'s requests to take the priority over client 1's requests.  Moreover, we would like to keep minimal the delay between requests and grants. This can be formalized by the \LTL formula $\phi$ given below where the signals in $I = \{r_1, r_2\}$ are controlled by the two clients, and the signals in $O = \{g_1,w_1, g_2,w_2 \}$ are controlled by the server. Moreover we add the next weight function $w : \Lit(O) \rightarrow \Z$:

\begin{minipage}{0.5\linewidth}
\begin{equation*} 
	\begin{split}
		\phi_1 &= \square(r_1 \rightarrow \nextLTL(w_1 \until g_1))\\
		\phi_2 &= \square(r_2 \rightarrow \nextLTL(w_2 \until g_2))\\
		\phi_3 &= \square(\lnot g_1 \vee \lnot g_2)\\
		\phi &= \phi_1 \wedge \phi_2 \wedge \phi_3
	\end{split} 
\end{equation*}
\end{minipage}
\begin{minipage}{0.4\linewidth}
\begin{center}
$w(l) =$
$\begin{cases} -1 \text{ if } l = w_1\\ -2 \text{ if } l = w_2\\ 0 \text{ otherwise.} \end{cases}$
\end{center}
\end{minipage}

\medskip\noindent
A possible strategy for the server is to behave as follows: it grants immediately any request of client~$2$ if the last ungranted request of client~$1$ has been emitted less than $n$ steps in the past, otherwise it grants the request of client~$1$. The mean-payoff value of this solution in the worst-case (when the two clients always emit their respective request) is equal to $-(1 + \frac{1}{n})$.
\end{example}

\paragraph{Reduction to safety games.} In~\cite{DBLP:journals/corr/abs-1210-3539,DBLP:conf/tacas/BohyBFR13}, we propose an antichain based algorithm for solving the $\LTLMP$ realizability and synthesis problems with a reduction to a two-player turn-based safety game. We here present this game $G$ without explaining the underlying reasoning, see~\cite{DBLP:journals/corr/abs-1210-3539} for more details. The game $G = (S_O, S_I, E, \alpha)$ is a turn-based safety game such that $S_O$ (resp. $S_I$) is the set of positions of Player $O$ (resp. Player $I$), $E$ is the set of edges labeled by $o \in \Sigma_O$ (resp. $i \in \Sigma_I$) when leaving a position in $S_O$ (resp. $S_I$), and $\alpha \subseteq S_O \cup S_I$ is the set of bad positions (i.e. positions that Player $O$ must avoid to reach). Let $\Win_O$ be the set of positions in $G$ from which Player $O$ can force Player $I$ to stay in $(S_O\cup S_I)\diff\alpha$, that is the set of winning positions for Player $O$. The safety game $G$ restricted to positions $\Win_O$ is a representation of a subset of the set of all winning strategies $\lambda_O$ for Player $O$ that ensure a value $\mpval(\Outcome(\lambda_O,\lambda_I))$ greater than or equal to the given threshold $\nu$, for all strategies $\lambda_I$ of Player $I$. Those strategies are called \textit{worst-case winning strategies}. 

Note that the reduction to safety games given in~\cite{DBLP:journals/corr/abs-1210-3539,DBLP:conf/tacas/BohyBFR13} allows to compute the set of all worst-case winning strategies (instead of a subset of them). 
Indeed the proposed algorithm is incremental on two parameters $K = 0, 1, \dots$ and $C= 0, 1, \dots$, and works as follows. For each value of $K$ and $C$, a corresponding safety game is constructed, whose number of states depends on $K$ and $C$. Those safety games have the following nice property. If player $O$ has a worst-case winning strategy in the safety game for $K = k$ and $C = c$, then he has a worst-case winning strategy in all the safety games for $K \geq k$ and $C \geq c$. The algorithm thus stops as soon as a worst-case winning strategy is found. There exist theoretical bounds $\mathbb{K}$ and $\mathbb{C}$ such that the set of all worst-case winning strategies can be represented by the safety game with parameters $\mathbb{K}$ and $\mathbb{C}$. However, $\mathbb{K}$ and $\mathbb{C}$ being huge, constructing this game is unfeasible in practice.

\paragraph{From safety games to MDPs.} We can go beyond $\LTLMP$ synthesis. Let $G$ be a safety game as above, that represents a subset of worst-case winning strategies. 
For each state $s \in \Win_O\cap S_O$, we denote by $\Sigma_{O,s} \subseteq \Sigma_O$ the set of actions that are safe to play in $s$ (i.e. actions that force Player $I$ to stay in $\Win_O$). For all $s \in \Win_O\cap S_O$, we know that $\Sigma_{O,s} \neq \emptyset$ by construction of $\Win_O$.
From this set of worst-case winning strategies, we want to compute the one that behaves \textit{the best against a stochastic opponent}. Let $\probI : \Sigma_I \rightarrow\ ]0,1]$ be a probability distribution on the actions of Player $I$. Note that we require $\support(\probI) = \Sigma_I$ so that it makes sense with the worst-case. 
By replacing Player $I$ by $\probI$ in the safety game $G$ restricted to $\Win_O$, we derive an MDP $M_G = (S, \ActionsO, \ActionsI, \edges, \distr)$ where:
\begin{itemize}
\itemsep0.1em
\item $S = \Win_O\cap S_O$,
\item $\ActionsO = \Sigma_O$, and for all $s \in S$, $\enabledactions = \Sigma_{O,s}$,
\item $\ActionsI = \Sigma_I$,
\item $\edges$, $\distr$ and $\reward$ are defined for all $s \in S$ and $\actionO \in \enabledactions$, such that:
\begin{itemize}
\item for all $\actionI \in \ActionsI$, $\edges(s, \actionO)(\actionI) = s'$ such that $(s, \actionO, s''), (s'', \actionI, s') \in E$,
\item for all $\actionI \in \ActionsI$, $\distr(s, \actionO)(\actionI) = \probI(\actionI)$, and
\item $\reward(s, \actionO) = w(\actionO)$.
\end{itemize}
\end{itemize}
Note that since $\Sigma_{O,s} \neq \emptyset$ for all $s \in S$, we have that $M$ is $\Sigma$-non-blocking. 

Computing the best strategy against a stochastic opponent among the worst-case winning strategies represented by $G$ reduces to solving the EMP problem for the MDP $M_G$\footnote{More precisely, it reduces to the EMP problem where the objective is to maximize the expected mean-payoff (see footnotes~\ref{fn:altobj} and~\ref{fn:MPmax}).}.

\begin{lemma}  \label{lem:LTLMono}
The MDP $M_G$ is monotonic, and functions $\distr, \reward$ are independent from $S$. 
\end{lemma}

\begin{proof}
It is shown in~\cite{DBLP:journals/corr/abs-1210-3539,DBLP:conf/tacas/BohyBFR13} that the safety game $G$ has properties of monotony. The set $S_O\cup S_I$ is equipped with a partial order $\preceq$ such that $(S_O\cup S_I, \preceq)$ is a complete lattice, and the sets $S_O, S_I$ and $\Win_O$ are closed for $\preceq$. For the MDP $M_G$ derived from $G$, we thus have that $(S, \preceq)$ is a (semi)lattice, and $S$ is closed for $\preceq$. Moreover, by construction of $G$ (see details in~\cite[Sec. 5.1]{DBLP:journals/corr/abs-1210-3539}) and $M_G$, 
we have that $\preceq$ is compatible with~$\edges$. By construction, $\distr, \reward$ are independent from $S$.
\qed\end{proof}

\paragraph{Symblicit algorithm.} 
In order to apply the pseudo-antichain based symblicit algorithm of Section~\ref{sec:paalgo}, Assumptions~\ref{atwo} and~\ref{aone} must hold for $M_G$. This is the case for Assumption~\ref{aone} since $\edges(s, \actionO)(\actionI)$ can be computed for all $s\in S, \actionO \in \enabledactions$ and $\actionI \in \ActionsI$ (see~\cite[Sec. 5.1]{DBLP:journals/corr/abs-1210-3539}), and $\distr$ is given by $\probI$. Moreover, from \cite[Prop. 24]{DBLP:journals/corr/abs-1210-3539} and $\support(\probI) = \Sigma_I$, we have an algorithm for computing $\lceil \Pre_{\actionO,\actionI}(\antclos \{x\})\rceil$, for all $x \in S$. So, Assumption~\ref{atwo} holds too. 
Notice also that for the MDP $M_G$ derived from the safety game $G$, the symbolic representations described in Section~\ref{subsec:symbrep} are compact, since $\distr$ and $\reward$ are independent from $S$ (see Lemma~\ref{lem:LTLMono}).

Therefore for this second class of MDPs, we have again an efficient pseudo-antichain based algorithm to solve the EMP problem, as indicated by the next experiments.

\paragraph{Experiments.} We have implemented the pseudo-antichain based symblicit algorithm for the EMP problem and integrated it into $\sf{Acacia+}$~($\sf{v2.2}$)~\cite{DBLP:conf/cav/BohyBFJR12}. $\sf{Acacia+}$ is a tool written in Python and C that provides an antichain based version of the algorithm described above for solving the $\LTLMP$ realizability and synthesis problems. The last version of $\sf{Acacia+}$ is available at \url{http://lit2.ulb.ac.be/acaciaplus/}, together with all the examples considered in this section. It can also be used directly online via a web interface. We compared our implementation with an MTBDD based symblicit algorithm implemented in $\sf{PRISM}$~\cite{prismEMP}. To the best of our knowledge, only strategy iteration algorithms are implemented for the EMP problem. In the sequel, for the sake of simplicity, we refer to the MTBDD based implementation as $\sf{PRISM}$ and to the pseudo-antichain based one as $\sf{Acacia+}$. Notice that for $\sf{Acacia+}$, the given execution times and memory consumptions only correspond to the part of the execution concerning the symblicit algorithm (and not the construction of the safety game $G$ and the subset of worst-case winning strategies that it represents).

\medskip
We compare the two implementations on a benchmark of~\cite{DBLP:conf/tacas/BohyBFR13} obtained from the $\LTLMP$ specification of Example~\ref{ex:ltl} extended with stochastic aspects ($\sf{Stochastic~shared~resource~arbiter}$). For the stochastic opponent, we set a probability distribution such that requests of client~$1$ are more likely to happen than requests of client~$2$: at each turn, client~$1$ has probability $\frac{3}{5}$ to make a request, while client~$2$ has probability $\frac{1}{5}$. The probability distribution $\probI : \Sigma_I \rightarrow\ ]0,1]$ is then defined as $\probI(\{\neg r_1, \neg r_2\}) = \frac{8}{25}$, $\probI(\{r_1, \neg r_2\}) = \frac{12}{25}$, $\probI(\{\neg r_1, r_2\}) = \frac{2}{25}$ and $\probI(\{r_1, r_2\}) = \frac{3}{25}$. We use the backward algorithm of $\sf{Acacia+}$ for solving the related safety games. The benchmark is parameterized in the threshold value $\nu$. Results are given in Table~\ref{table:LTL}. Note that the number of states in the MDPs depends on the implementation. Indeed, for $\sf{PRISM}$, it is the number of reachable states of the MDP, denoted $|M_G^R|$, that is, the states that are really taken into account by the algorithm, while for $\sf{Acacia+}$, it is the total number of states since unlike $\sf{PRISM}$, our implementation does not prune unreachable states. For this application scenario, we observe that the ratio (number of reachable states)/(total number of states) is in general quite small\footnote{For all the MDPs considered in Tables~\ref{table:STRIPS1} and \ref{table:STRIPS2}, this ratio is $1$.}. 

\begin{table}[h!]
	\caption{Expected mean-payoff on the $\sf{Stochastic\ shared\ resource\ arbiter}$ benchmark with $2$ clients and decreasing threshold values. The column $\nu$ gives the threshold, $|M_G^R|$ the number of reachable states in the MDP, and all other columns have the same meaning as in Table~\ref{table:STRIPS1}. The expected mean-payoff $\EMP_\lambda$ of the optimal strategy $\lambda$ for all the examples is $-0.130435$.}	\label{table:LTL}
	\centering
	\scriptsize
 	\begin{tabular}{|r||r|r|r|r|r|r|r|r||r|r|r|r|r|r|r|r|r|r|r|r|r|r|r|r|}
		\hline
	  	& \multicolumn{8}{|c||}{{\small $\sf{Acacia+}$}} & \multicolumn{5}{|c|}{{\small $\sf{PRISM}$}}\rule[-2pt]{0pt}{10pt}\\
		$\nu~$ & $|M_G|$  & \ $\#$it \  & $\ |\partlump|\ $ & \ lump  \ &  \ LS \  &  \ impr  \ & \  total \  & \ mem \ & $|M_G^R| $ & \ constr \  & \ strat \  & \ total \ & \ mem\ \rule[-3pt]{0pt}{10pt}\\
\hline
$\ -1.1\ $ & $\ 5259\ $ & $\ 2\ $ & $\ 22\ $ & $\ 0.12\ $ & $\ 0.01\ $ & $\ 0.02\ $  & $\ 0.15\ $ & $\ 17.4\ $ & $\ 691\ $ & $\ 0.43\ $ & $\ 0.07\ $  & $\ 0.50\ $ & $\ 168.1\ $\rule[-3pt]{0pt}{10pt}\\
$\ -1.05\ $ & $\ 72159\ $ & $\ 2\ $ & $\ 42\ $ & $\ 0.86\ $ & $\ 0.01\ $ & $\ 0.09\ $  & $\ 0.97\ $ & $\ 17.7\ $ & $\ 6440\ $ & $\ 1.58\ $ & $\ 0.18\ $  & $\ 1.76\ $ & $\ 249.9\ $\rule[-3pt]{0pt}{10pt}\\
$\ -1.04\ $ & $\ 35750\ $ & $\ 2\ $ & $\ 52\ $ & $\ 1.63\ $ & $\ 0.02\ $ & $\ 0.13\ $  & $\ 1.79\ $ & $\ 18.1\ $ & $\ 3325\ $ & $\ 1.78\ $ & $\ 0.28\ $  & $\ 2.06\ $ & $\ 264.1\ $\rule[-3pt]{0pt}{10pt}\\
$\ -1.03\ $ & $\ 501211\ $ & $\ 2\ $ & $\ 70\ $ & $\ 4.41\ $ & $\ 0.04\ $ & $\ 0.26\ $  & $\ 4.71\ $ & $\ 18.8\ $ & $\ 15829\ $ & $\ 4.83\ $ & $\ 0.46\ $  & $\ 5.29\ $ & $\ 277.0\ $\rule[-3pt]{0pt}{10pt}\\
$\ -1.02\ $ & $\ 530299\ $ & $\ 2\ $ & $\ 102\ $ & $\ 16.62\ $ & $\ 0.11\ $ & $\ 0.64\ $  & $\ 17.39\ $ & $\ 20.2\ $ & $\ 11641\ $ & $\ 6.74\ $ & $\ 0.59\ $  & $\ 7.33\ $ & $\ 343.4\ $\rule[-3pt]{0pt}{10pt}\\
$\ -1.01\ $ & $\ 4120599\ $ & $\ 2\ $ & $\ 202\ $ & $\ 237.78\ $ & $\ 0.50\ $ & $\ 3.94\ $  & $\ 242.30\ $ & $\ 26.2\ $ & $\ 43891\ $ & $\ 29.91\ $ & $\ 1.61\ $  & $\ 31.52\ $ & $\ 642.5\ $\rule[-3pt]{0pt}{10pt}\\
$\ -1.005\ $ & $\ 64801599\ $ & $\ 2\ $ & $\ 402\ $ & $\ 3078.88\ $ & $\ 3.07\ $ & $\ 28.19\ $  & $\ 3110.50\ $ & $\ 48.0\ $ & $\ 563585\ $ & $\ 179.23\ $ & $\ 4.72\ $  & $\ 183.95\ $ & $\ 1629.2\ $\rule[-3pt]{0pt}{10pt}\\
$\ -1.004\ $ & $\ 63251499\ $ & $\ 2\ $ & $\ 502\ $ & $\ 7357.72\ $ & $\ 5.68\ $ & $\ 52.81\ $  & $\ 7416.77\ $ & $\ 60.5\ $ & $\ 264391\ $ & $\ 270.30\ $ & $\ 7.71\ $  & $\ 278.01\ $ & $\ 2544.0\ $\rule[-3pt]{0pt}{10pt}\\
$\ -1.003\ $ & $\ 450012211\ $ & $\ 2\ $ & $\ 670\ $ & $\ 23455.44\ $ & $\ 12.72\ $ & $\ 120.25\ $  & $\ 23589.49\ $ & $\ 93.6\ $ & $\ \ $ & $\ \ $ & $\ \ $  & $\ \ $ & $\ \memout\ $\rule[-3pt]{0pt}{10pt}\\
\hline
	\end{tabular}
	\normalsize
\end{table}

On this benchmark, $\sf{PRISM}$ is faster that $\sf{Acacia+}$ on large models, but $\sf{Acacia+}$ is more efficient regarding the memory consumption and this in spite of considering the whole state space. For instance, the last MDP of Table~\ref{table:LTL} contains more than $450$ millions of states and is solved by $\sf{Acacia+}$ in around $6.5$ hours with less than $100$Mo of memory, while for this example, $\sf{PRISM}$ runs out of memory. Note that the surprisingly large amount of memory consumption of both implementations on small instances is due to Python libraries loaded in memory for $\sf{Acacia+}$, and to the JVM and the CUDD package for $\sf{PRISM}$~\cite{DBLP:conf/hvc/JansenKOSZ07}.

To fairly compare the two implementations, let us consider Figure~\ref{fig:time} (resp. Figure~\ref{fig:memory}) that gives a graphical representation of the execution times (resp.  the memory consumption) of $\sf{Acacia+}$ and $\sf{PRISM}$ as a function of the number of states taken into account, that is, the total number of states for $\sf{Acacia+}$ and the number of reachable states for $\sf{PRISM}$. For that experiment, we consider the benchmark of examples of Table~\ref{table:LTL} with four different probability distributions on $\Sigma_I$. Moreover, for each instance, we consider the two MDPs obtained with the backward and the forward algorithms of $\sf{Acacia+}$ for solving safety games. The forward algorithm always leads to smaller MDPs. On the whole benchmark, $\sf{Acacia+}$ times out on three instances, while $\sf{PRISM}$ runs out of memory on four of them. Note that all scales in Figures~\ref{fig:time} and~\ref{fig:memory} are logarithmic.

\begin{figure}[!h]
   \begin{minipage}[t]{.48\linewidth}
	\includegraphics[width=\textwidth]{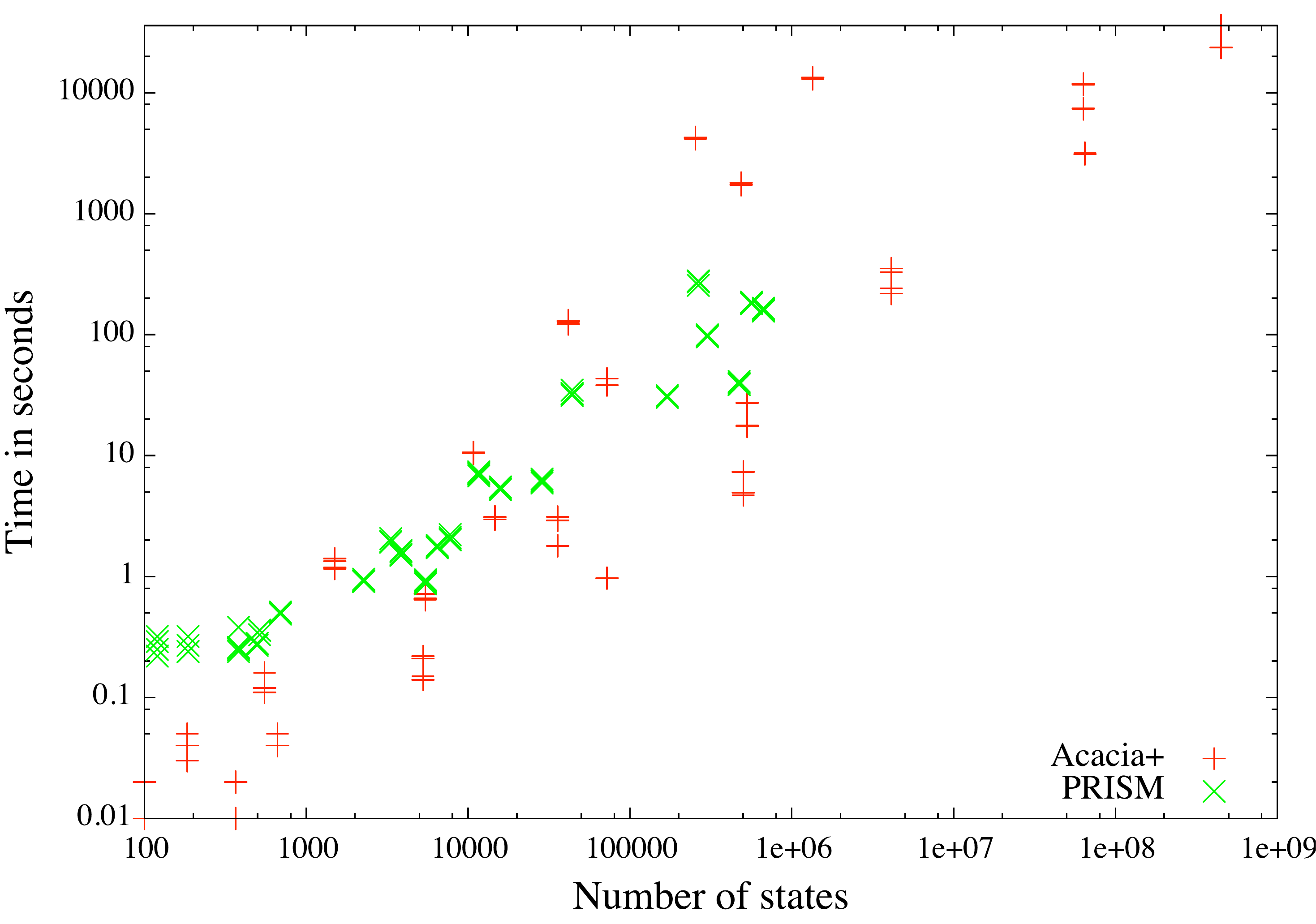}
	\caption{Execution time}
	\label{fig:time}   
\end{minipage} \hfill
   \begin{minipage}[t]{.48\linewidth}
	\includegraphics[width=\textwidth]{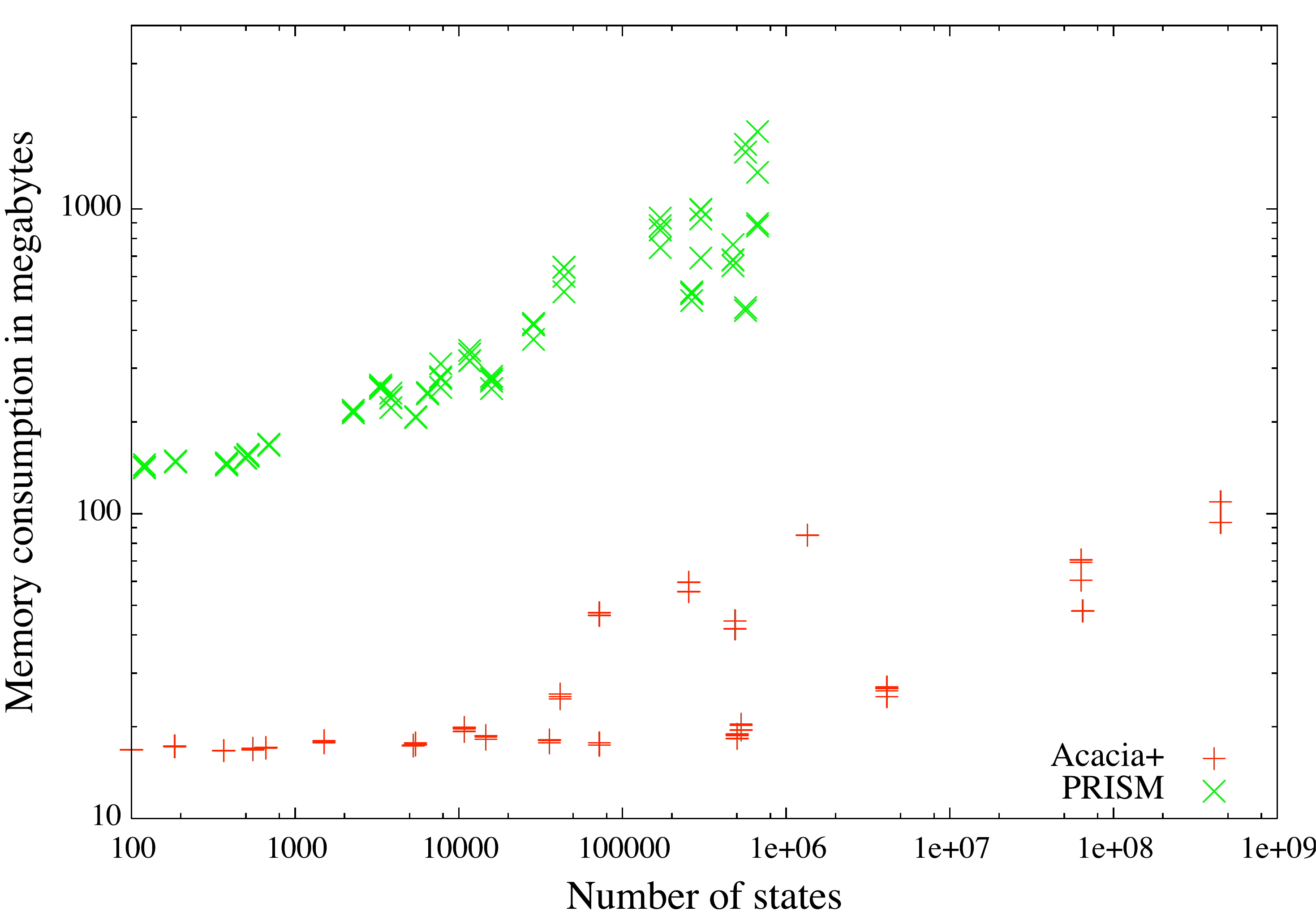}
	\caption{Memory consumption}
	\label{fig:memory}  
 \end{minipage}
\end{figure}

On Figure~\ref{fig:time}, we observe that for most of the executions, $\sf{Acacia+}$ works faster that $\sf{PRISM}$. We also observe that  $\sf{Acacia+}$ does not behave well for a few particular executions, and that these executions all correspond to MDPs obtained from the forward algorithm of $\sf{Acacia+}$.

Figure~\ref{fig:memory} shows that regarding the memory consumption, $\sf{Acacia+}$ is  more efficient than $\sf{PRISM}$ and it can thus solve larger MDPs (the largest MDP solved by $\sf{PRISM}$ contains half a million states while $\sf{Acacia+}$ solves MDPs of more than $450$ million states). This points out that monotonic MDPs are better handled by pseudo-antichains, which exploit the partial order on the state space, than by BDDs.

\medskip
Finally, in the majority of experiments we performed for both the EMP and the SSP problems, we observe that most of the execution time of the pseudo-antichain based symblicit algorithms is spent for lumping. It is also the case for the MTBDD based symblicit algorithm~\cite{DBLP:conf/qest/WimmerBBHCHDT10}.

\section{Conclusion} \label{sec:conclusion}
In this paper, we have presented the interesting class of monotonic MDPs, and the new data structure of pseudo-antichain. We have shown how monotonic MDPs can be exploited by symblicit algorithms using pseudo-antichains (instead of MTBDDs) for two quantitative settings: the expected mean-payoff and the stochastic shortest path. Those algorithms have been implemented, and we have reported promising experimental results for two applications coming from automated planning and \LTLMP synthesis. We are convinced that pseudo-antichains can be used in the design of efficient algorithms in other contexts like for instance model-checking or synthesis of non-stochastic models, as soon as a natural partial order can be exploited.

\medskip
\textbf{Acknowledgments. } We would like to thank Mickael Randour for his fruitful discussions, Marta Kwiatkowska, David Parker and Christian Von Essen for their help regarding the tool $\sf{PRISM}$, and Holger Hermanns and Ernst Moritz Hahn for sharing with us their prototypical implementation.

\bibliographystyle{abbrv}
\bibliography{biblio}

\end{document}